\documentclass{new_tlp}

\usepackage[utf8]{inputenc}
\usepackage[T1]{fontenc}
\usepackage[english]{babel}

\usepackage{amsmath,amssymb}
\usepackage{thmtools}
\usepackage{thm-restate}
\usepackage{wasysym}
\usepackage[final]{microtype}
\usepackage{scalerel}
\usepackage{lmodern}

\usepackage{graphicx}
\usepackage{xspace}
\usepackage{csquotes}
\usepackage{booktabs}
\usepackage{bussproofs}
\usepackage{MnSymbol}

{\bfseries}{\itshape}

\newlength{\dx}
\newlength{\dy}

\usepackage{mathtools}

\usepackage{tikz}
\usetikzlibrary{arrows,automata,positioning,backgrounds,shapes, calc,decorations.pathreplacing,decorations.pathmorphing,decorations.markings}

\usepackage[shortlabels]{enumitem}
\setlist{leftmargin=*}

\usepackage[final,hidelinks]{hyperref}
\usepackage{cleveref}
\usepackage[strings]{underscore}

\newcommand{\do}{}

\usepackage[textwidth=1.5cm,textsize=footnotesize,prependcaption,colorinlistoftodos]{todonotes}

\binoppenalty=\maxdimen
\relpenalty=\maxdimen

\renewcommand{\paragraph}[1]{ \vspace{1ex} \noindent \textbf{#1}}

\newtheorem{theorem}{Theorem}
\newtheorem{lemma}[theorem]{Lemma}

\newtheorem{definition}[theorem]{Definition}
\newtheorem{example}[theorem]{Example}

\definecolor{mLightBrown}{HTML}{EB811B}

\def\define#1#2#3%
{%
\renewcommand*{\do}[1]{%
 \expandafter\newcommand\csname
 #1\endcsname{#2}
}
\docsvlist{#3}
}

\define{#1mc}
{{\ensuremath{\mathcal{#1}}}\xspace}
{A,B,C,D,E,F,G,H,I,J,K,L,M,N,O,P,Q,R,S,T,U,V,W,X,Y,Z}

\define{#1mf}
{{\ensuremath{\mathfrak{#1}}}\xspace}
{A,B,C,D,E,F,G,H,I,J,K,L,M,N,O,P,Q,R,S,T,U,V,W,X,Y,Z}

\define{#1sf}
{{\ensuremath{\mathsf{#1}}}\xspace}
{A,B,C,D,E,F,G,H,I,J,K,L,M,N,O,P,Q,R,S,T,U,V,W,X,Y,Z}

\define{#1bb}
{{\ensuremath{\mathbb{#1}}}\xspace}
{A,C,D,E,F,G,H,I,J,K,L,M,N,O,P,Q,R,S,T,U,V,W,X,Y,Z}

\define{#1tt}
{{\ensuremath{\mathtt{#1}}}\xspace}
{A,C,D,E,F,G,H,I,J,K,L,M,N,O,P,Q,R,S,T,U,V,W,X,Y,Z}

\define{#1bf}
{{\ensuremath{\mathbf{#1}}}\xspace}
{a,x,y,z,X}

\define{N#1}
{{\ensuremath{\mathbf{#1}}}\xspace}
{C,R,I,V,T}

\define{#1}
{{\textsc{#1}}\xspace}
{NP,coNP,PSpace,ExpSpace}

\newcommand{\PTime}{\textsc{P}\xspace}

\newcommand{\study}[1]{\url{https://clinicaltrials.gov/ct2/show/#1}\xspace}

\newcommand{\EL}{\ensuremath{\mathcal{E\kern-.25ex L}}\xspace}
\newcommand{\ELH}{\ensuremath{\mathcal{E\kern-.25ex LH}}\xspace}
\newcommand{\ELHb}{\ensuremath{\mathcal{E\kern-.25ex LH}_\bot}\xspace}
\newcommand{\bELHb}{\ensuremath{\mathcal{E\kern-.25ex LH}_\bot}\xspace}

\newcommand{\bit}{\ensuremath{\mathtt{bit}}}
\newcommand{\sign}{\ensuremath{\mathtt{sign}}}
\newcommand{\ovf}{\ensuremath{\mathtt{ovf}}}

\newcommand{\aboxf}{\ensuremath{\mathtt{\Imf_\Amc^{\mathrm{fin}}}}} 
\newcommand{\aboxi}{\ensuremath{\mathtt{\Imf_\Amc}}} 
\newcommand{\rew}[1]{\ensuremath{#1_\Tmc}}

\newcommand{\repr}[2]{\ensuremath{\mathtt{rep}^{#1}(#2)}\xspace}
\newcommand{\trans}[4][\xbf]{\ensuremath{[#2]^{#3}(#1,#4)}\xspace}

\newcommand{\MFOTL}{MFOTL\xspace}

\usepackage{framed}
\definecolor{shadecolor}{RGB}{237, 237, 171}

\newcommand{\ie}{i.e.\ }
\newcommand{\eg}{e.g.\ }
\newcommand{\wrt}{w.r.t.\ }

\newcommand{\qed}{\mathproofbox}

\newcommand{\BreastCancerPatientIdx}{1}
\newcommand{\SkinAndBreastPatientIdx}{2}
\newcommand{\SkinOfBreastPatientIdx}{3}
\newcommand{\SkinAndBreastPatient}{\ensuremath{p_\SkinAndBreastPatientIdx}\xspace}
\newcommand{\BreastCancerPatient}{\ensuremath{p_\BreastCancerPatientIdx}\xspace}

\newcommand{\SkinOfBreastPatient}{\ensuremath{p_\SkinOfBreastPatientIdx}\xspace}

\newcommand{\findingSite}{findingSite\xspace}
\newcommand{\Structure}{Structure\xspace}

\newcommand{\diag}{\ensuremath{\text{diagnosedWith}}\xspace}

\newcommand{\ChemoP}{\ensuremath{\text{ChemotherapyPatient}}\xspace}

\newcommand{\CancerP}{\ensuremath{\text{CancerPatient}}\xspace}

\newcommand{\ex}[1]{\text{#1}\xspace}

\newcommand{\lnext}{{\ocircle}}
\newcommand{\lprev}{{\ocircle^-}}

\newcommand{\leventually}{\scaleobj{1.2}{\lozenge}}

\newcommand{\lalways}{\Box}
\newcommand{\luntil}{\,\mathcal{U}}

\newcommand{\lsince}{\,\mathcal{S}}

\renewcommand{\diamond}[1][]{\ensuremath{\hspace*{.3em}\clap{\upshape\fontsize{5.5}{6.5}\selectfont\text{#1}}\clap{\leventually}\hspace*{.3em}}\xspace}

\newcommand{\diamondp}{\diamond[\raisebox{.67ex}{$\boldsymbol+$}]}
\newcommand{\diamondm}{\diamond[\raisebox{.67ex}{$\boldsymbol-$}]}
\newcommand{\diamondpm}{\diamond[\raisebox{.9ex}{$\boldsymbol\pm$}]}
\newcommand{\diamondc}[1][]{\ensuremath{\diamond[\raisebox{.77ex}{$\hskip-.05em\rlap{$\mathsf{c}$}\hskip-0.01em\mathsf{c}$}]\ifthenelse{\equal{#1}{}}{}{_{#1}}}\xspace}
\newcommand{\da}{\diamond[\raisebox{.65ex}{$\boldsymbol\star$}]}
\newcommand{\db}{\diamond[\raisebox{.65ex}{$\bullet$}]}
\newcommand{\dc}{\diamond[\raisebox{.55ex}{$\boldsymbol\dagger$}]}
\newcommand{\dd}{\diamond[\raisebox{.6ex}{$\boldsymbol\ddagger$}]}

\newcommand{\Ds}{\ensuremath{\Dmf}\xspace}
\newcommand{\Dpm}{\ensuremath{\Dmf^\pm}\xspace}
\newcommand{\Dc}{\ensuremath{\Dmf^\mathsf{c}}\xspace}

\newcommand{\id}{\ensuremath{\mathrm{id}_{2^\Zbb}}\xspace}

\newcommand{\sdiamond}[1][]{\ensuremath{\smash{\hspace*{.3em}\clap{\fontsize{5}{6}\selectfont\text{#1}}\clap{$\lozenge$}\hspace*{.3em}}}\xspace}
\newcommand{\sdiamondc}{\sdiamond[\raisebox{.65ex}{$\hskip-0.05em\mathsf{c}$}]}
\newcommand{\bsdiamondc}{\ensuremath{\smash{\hspace*{.3em}\clap{\fontsize{5}{6}\selectfont\text{\raisebox{.8ex}{$\hskip-0.05em\mathsf{c}$}}}\clap{$\boldsymbol\lozenge$}\hspace*{.3em}}}\xspace}

\binoppenalty=\maxdimen

\newcommand{\LTLbin}{\ensuremath{\text{LTL}^{\text{bin}}}\xspace}
\newcommand{\ELHbLTL}{\ensuremath{\mathcal{E\kern-.25ex LH}_\bot}-\LTLbin\xspace}
\newcommand{\TELdinfl}{\ensuremath{\mathcal{T\kern-.25ex E\kern-.25ex L^{\sdiamond}_{\mathsf{infl}}}}\xspace}

\newcommand{\dELHb}{\ensuremath{\mathcal{T}\kern-.25ex\ELHb^{\sdiamondc\!,\mathsf{lhs}}}\xspace}

\newcommand{\bdELHb}{\ensuremath{\mathbcal{T}\kern-.25ex\bELHb^{\bsdiamondc\!,\mathsf{lhs}}}\xspace}

\newcommand{\dELHbm}{\ensuremath{\mathcal{T}\kern-.25ex\ELHb^{\sdiamondc\!,\mathsf{lhs},-}}\xspace}

\DeclareMathOperator{\Var}{Var}

\DeclareMathOperator{\Ind}{Ind}
\DeclareMathOperator{\ans}{ans}
\DeclareMathOperator{\cert}{cert}
\DeclareMathOperator{\tem}{tem(\Amc)}
\DeclareMathOperator{\zaplus}{rep(\Amc)}

\DeclareMathOperator{\rewo}{rew}
\DeclareMathOperator{\mwa}{mwa}

\DeclareMathOperator{\PRED}{\mathsf{Pred}}
\DeclareMathOperator{\POS}{\mathsf{Pos}}
\DeclareMathOperator{\NEG}{\mathsf{Neg}}

\tikzset{%
  ->,>=stealth',shorten >=1pt,
  state/.style={draw=none,minimum size=20pt,inner sep=1pt,fill=black!25},
  query-inner/.style={rectangle, fill=black!25,rounded corners=0.1cm},
  query-outer/.style={draw=darkgreen,thick,rounded corners=.55cm,inner sep=2pt},
  named/.style={state,query-inner},
  anonymous/.style={state,circle},
  label/.style={fill=none,minimum size=17pt,inner sep=0pt, font=\scriptsize},
  diagnosedWith/.style={draw=diagnosedWith},
  findingSite/.style={draw=findingSite},
  invisible/.style={opacity=0},
  visible on/.style={alt=#1{}{invisible}},
  alt/.code args={<#1>#2#3}{%
    \alt<#1>{\pgfkeysalso{#2}}{\pgfkeysalso{#3}} 
  },
}

\title[Temporal Minimal-World Semantics for Sparse ABoxes]{Temporal Minimal-World Query Answering over Sparse ABoxes}

\author[Stefan Borgwardt, Walter Forkel, Alisa Kovtunova]{STEFAN BORGWARDT and WALTER FORKEL and ALISA KOVTUNOVA\\
Chair for Automata Theory, Technische Universität Dresden, Germany\\
\email{firstname.lastname@tu-dresden.de}}

\begin{document}


\maketitle

\begin{abstract}
  Ontology-mediated query answering is a popular paradigm for enriching answers to user queries with background knowledge.
  For querying the \emph{absence} of information, however, there exist only few ontology-based approaches.
  Moreover, these proposals conflate the closed-domain and closed-world assumption, and therefore are not suited to deal with the anonymous objects that are common in ontological reasoning.
  %
  Many real-world applications, like processing electronic health records (EHRs), also contain a temporal dimension, and require efficient reasoning algorithms.
  Moreover, since medical data is not recorded on a regular basis, reasoners must deal with sparse data with potentially large temporal gaps.
  %

  Our contribution consists of two main parts:
  In the first part we introduce a new closed-world semantics for answering conjunctive queries with negation over ontologies formulated in the description logic \ELHb, which is based on the \emph{minimal} canonical model.
  We propose a rewriting strategy for dealing with negated query atoms, which shows that query answering is possible in polynomial time in data complexity.
  In the second part, we extend this minimal-world semantics for answering metric temporal conjunctive queries with negation over the lightweight temporal logic \dELHbm and obtain similar rewritability and complexity results.
  This paper is under consideration in Theory and Practice of Logic Programming (TPLP).
\end{abstract}

\section{Introduction}

\emph{Ontology-mediated query answering (OMQA)} allows using background knowledge for answering user queries, supporting data-focused applications offering search, analytics, or data integration functionality.
An \emph{ontology} is a logical theory formulated in a decidable fragment of first-order logic, \eg a description logic (DL)~\cite{BHLS-17}, with a trade-off between the expressivity of the ontology language and the efficiency of query answering.
\emph{Rewritability} is a popular topic of research, the idea being to reformulate ontological queries, such as conjunctive queries (CQs) enhanced by an ontology, into database queries that can be answered by traditional database management systems~\cite{MuTh-RW14,EOS+-AAAI12,bienvenu2015ontology,CCK+-SWJ17,KHS+-JWS17}.

Ontology-based systems do not use the \emph{closed-domain} and \emph{closed-world} semantics of databases.
Instead, they acknowledge that unknown (\emph{anonymous}) objects may exist (\emph{open domain}) and that facts that are not explicitly stated may still be true (\emph{open world}).
Anonymous objects are related to \emph{null} values in databases; for example, if we know that every person has a mother, then first-order models include all mothers, even though they may not be mentioned explicitly in the input dataset.
In addition, the open-world assumption ensures that, if the dataset does not contain an entry on, \eg whether a person is male or female, then we do not infer that this person is neither male nor female, but rather consider all possibilities.

In the last decade, OMQA has been extended to temporal description logics (DLs) that combine terminological and temporal knowledge representation capabilities~\cite{WoZa-FroCoS98,LuWZ-TIME08,artale2017ontology}.
To obtain tractable reasoning procedures, lightweight temporal DLs have been developed \cite{AKL+-TIME07,GuJK-IJCAI16}.
The idea is to use temporal operators, often from the linear temporal logic LTL, inside DL axioms.
For example, $\diamondm\exists\ex{diagnosis}.\ex{BrokenLeg}\sqsubseteq\exists\ex{treatment}.\ex{LegCast}$ states that after breaking a leg one has to wear a cast.
However, this basic approach cannot represent the distance of events, \eg that the cast only has to be worn for a fixed amount of time.
Recently, metric temporal ontology languages have been investigated \cite{GuJO-ECAI16,BBK+-FroCoS17,BKR+-JAIR18}, which allow to replace~$\diamondm$ in the above axiom with~$\diamond_{[-8,0]}$, \ie wearing the cast is required only if the leg was broken $\le8$ time points (\eg weeks) ago.

The biomedical domain is a fruitful area for OMQA methods, due to the availability of large ontologies covering a multitude of topics%
\footnote{https://bioportal.bioontology.org}
and the demand for managing large amounts of patient data, in the form of \emph{electronic health records (EHRs)} \cite{SCWB-17}.
For example, for the preparation of clinical trials%
\footnote{https://clinicaltrials.gov}
a large number of patients need to be screened for eligibility, and an important area of current research is how to automate this process \cite{patel2007matching,besana2010using,KoPr-JMIR14,TAC+-14,NWP+-MIDM15}.

In particular, many clinical trials contain temporal eligibility criteria~\cite{crowe2015designing}, such as:

``type 1 diabetes with duration at least 12 months''\footnote{\study{NCT02280564}};
``known history of heart disease or heart rhythm abnormalities''\footnote{\study{NCT02873052}};

``CD4+ lymphocytes count > 250/mm3, for at least 6 months''\footnote{\study{NCT02157311}};
or
``symptomatic recurrent paroxysmal atrial fibrillation (PAF) (> 2 episodes in the last 6 months)''\footnote{\study{NCT00969735}}.
Moreover, measurements, diagnoses, and treatments in a patients' EHR are clearly valid only for a certain amount of time.
To automatically screen patients according to the temporal criteria above, one needs a sufficiently powerful formalism that can reason about biomedical and temporal knowledge.

Additionally, ontologies and EHRs mostly contain \emph{positive} information, while clinical trials also require certain \emph{exclusion criteria} to be absent in the patients.
For example, we may want to select only patients that have \emph{not} been diagnosed with cancer,%
\footnote{An exclusion criterion in https://clinicaltrials.gov/ct2/show/NCT01463215}
but such information cannot be entailed from the given knowledge.
The culprit for this problem is the open-world semantics, which considers a cancer diagnosis possible unless it has been explicitly ruled out.
Unfortunately, existing approaches like (partial) closed-world semantics \cite{LuSW-IJCAI13,AhOS-IJCAI16} or epistemic logics are unable to deal with closed-world knowledge over anonymous objects \cite{wolter2000first,calvanese2006epistemic}.

In the first part of this paper, we introduce a new closed-world semantics to answer \emph{conjunctive queries with (guarded) negation} \cite{BatO-VLDB12} over ontologies formulated in \ELHb, an ontology language that covers many biomedical ontologies, \eg SNOMED\,CT.%
\footnote{\url{https://www.snomed.org/}}
Our semantics, called \emph{minimal-world semantics}, is based on the \emph{minimal canonical model}, which encodes all inferences of the ontology in the most concise way possible.
As a side effect, this means that ordinary CQs without negation are interpreted under the standard open-world semantics.
In order to properly handle negative knowledge about anonymous objects, however, we have to be careful in the construction of the canonical model, in particular about the number and types of anonymous objects that are introduced.
Since in general the minimal canonical model is infinite, we develop a rewriting technique, in the spirit of the combined approach of \cite{lutz2009conjunctive,KLT+-IJCAI11}, and most closely inspired by \cite{EOS+-AAAI12,bienvenu2015ontology}, which allows us to evaluate conjunctive queries with negation over a finite part of the canonical model, using traditional database techniques.

In the second part of the paper, we extend the minimal-world semantics to support also temporal information.

When working with EHRs, which contain information for specific points in time only, it is especially important to be able to infer what happened to the patient in the meantime.
For example, if a patient is diagnosed with a (currently) incurable disease like \ex{Diabetes}, they will still have the disease at any future point in time.
Similarly, if the EHR contains two entries of \ex{CD4Above250} four weeks apart, one may reasonably infer that this was true for the whole four weeks.

We use the qualitative temporal DL \dELHbm \cite{BoFoKo-RRML} that can express the former statement by declaring \ex{Diabetes} as \emph{expanding} via the axiom $\diamondm\ex{Diabetes}\sqsubseteq\ex{Diabetes}$.
Additionally, by a special kind of metric temporal operators it is allowed to write $\diamondc[4]\ex{CD4Above250}\sqsubseteq\ex{CD4Above250}$, making the measurement \emph{convex} for a specified length of time~$n$ (\eg $4$ weeks).
This means that information is interpolated between time points of distance less than~$n$, thereby computing a convex closure of the available information.
The threshold~$n$ allows to distinguish the case where two mentions of \ex{CD4Above250} are years apart, and are therefore unrelated.
Reasoning in \dELHbm is tractable in data complexity, because $\diamond$-operators are only allowed on the left-hand side of concept inclusions \cite{BoFoKo-RRML,GuJK-IJCAI16}, which is also common for temporal DLs based on \textit{DL-Lite} \cite{AKWZ-IJCAI13,AKK+-IJCAI15}.

Here we consider the problem of answering metric temporal queries over \dELHbm knowledge bases with the minimal-world semantics introduced in the first part of the paper.
Our query language extends the temporal conjunctive queries from~\cite{BaBL-JWS15} by metric temporal operators~\cite{GuJO-ECAI16,BBK+-FroCoS17} and negation.
For example, we can use queries like $\lalways_{[-12,0]}(\exists y.\diag(x,y) \land \ex{Diabetes}(y)) $ to select all patients~$x$ for which the first criterion from above is satisfied.

By extending our combined rewriting approach, we show that the data complexity of temporal query answering is not higher than for (positive) atemporal queries in \ELHb, and also provide a tight combined complexity result of \ExpSpace.
Unlike most research on temporal query answering~\cite{BaBL-JWS15,AKK+-IJCAI15}, we do not assume that input data is given for all time points in a certain interval, but rather at sporadic time points with arbitrarily large gaps~\cite{BKR+-JAIR18}.
The main technical difficulty is to determine which additional time points are relevant for answering a query, and how to access these time points without having to fill infinitely many gaps.

This paper extends the conference paper~\cite{BoFo-JELIA19} by allowing also non-rooted queries and extends parts of~\cite{BoFoKo-RRML} by providing full proofs of all results and improving some of the running examples.

\section{Preliminaries}
\label{sec:preliminaries}

We recall the definitions of \ELHb and first-order queries, which are needed for our rewriting of conjunctive queries with negation.

\paragraph{The Description Logic \bELHb.}
Let $\NC, \NR, \NI$ be countably infinite sets of \emph{concept}, \emph{role}, and \emph{individual names}, respectively.
A \emph{concept} is built according to the syntax rule
$ C ::= A \mid \top \mid \bot \mid C \sqcap C \mid \exists r.C $,
where $A \in \NC$ and $r \in \NR$.
An \emph{ABox} is a finite set of \emph{concept assertions} $A(a)$ and \emph{role assertions} $r(a,b)$, where $a,b \in \NI$.
A \emph{TBox} is a finite set of \emph{concept inclusions} $C \sqsubseteq D$ and \emph{role inclusions} $r \sqsubseteq s$, where $C,D$ are concepts and $r,s$ are roles.
In the following we assume the TBox to be in normal form, \ie that it contains only inclusions of the form
\[
  A_1 \sqcap \dots \sqcap A_n \sqsubseteq B,
  \qquad
  A \sqsubseteq \exists r.B,
  \qquad
  \exists r.A \sqsubseteq B,
  \qquad
  r \sqsubseteq s
\]
where $A_{(i)} \in \NC \cup \{ \top \}$, $B \in \NC \cup \{ \bot \}$, $r,s \in \NR$, and $n \geq 1$.

A \emph{knowledge base (KB)} (or \emph{ontology}) $\Kmc = \Tmc \cup \Amc$ consists of a TBox~$\Tmc$ and an ABox~$\Amc$.
We refer to the set of individual names occurring in $\Kmc$ by $\Ind(\Kmc)$.
We write $C\equiv D$ to abbreviate the two inclusions $C\sqsubseteq D$, $D\sqsubseteq C$, and similarly for role inclusions.

The semantics of \ELHb is defined as usual \cite{BCM+-07} in terms of \emph{interpretations} $\Imc = (\Delta^\Imc, \cdot^\Imc)$, where the \emph{domain} $\Delta^\Imc$ is a non-empty set and the functions~$\cdot^{\Imc}$ are extended from concept and role names inductively as follows:
\begin{align*}
  \top^{\Imc} &:= \Delta^\Imc
  &
  \bot^{\Imc} &:= \emptyset
  \\
  (C\sqcap D)^{\Imc} &:= C^{\Imc}\cap D^{\Imc}
  &
  (\exists r.C)^{\Imc} &:= \big\{ d\in\Delta^\Imc \mid \exists e\in C^{\Imc}\colon (d,e)\in r^{\Imc} \big\}
\end{align*}
$\Imc$ is a \emph{model} of (or \emph{satisfies})

  a concept inclusion $C\sqsubseteq D$ if $C^{\Imc}\subseteq D^{\Imc}$ holds,
  a role inclusion $r\sqsubseteq s$ if $r^{\Imc}\subseteq s^{\Imc}$ holds,
  a concept assertion $A(a)$ if $a\in A^{\Imc}$,
  a role assertion $r(a,b)$ if $(a,b)\in r^{\Imc}$,
  and the KB~\Kmc if it satisfies all axioms in~\Kmc.

In the following, we assume all KBs to be consistent and
make the standard name assumption, \ie that for every individual name~$a$ in any interpretation $\Imc$ we have $a^\Imc = a$.
An axiom~$\alpha$ is \emph{entailed} by~\Kmc (written $\Kmc \models \alpha$) if $\alpha$ is satisfied in all models of~$\Kmc$.
We abbreviate $\Kmc\models C\sqsubseteq D$ to $C\sqsubseteq_\Tmc D$, and similarly for role inclusions; note that the ABox does not influence the entailment of inclusions.
Entailment in \ELHb can be decided in polynomial time \cite{BaBL-IJCAI05}.

\paragraph{The Temporal Description Logic \dELHbm.}
\label{sec:delhb}
We recall the relevant definitions from~\cite{BoFoKo-RRML}, starting with the MTL operators that are used in \dELHbm.

LTL formulas are formulated over a finite set~$P$ of \emph{propositional variables}. In this section, we consider only formulas built according to the syntax rule $\varphi ::= p \mid \varphi\land\varphi \mid \varphi\lor\varphi \mid \diamond_I\varphi$, where $p\in P$ and $I$ is an interval in~\Zbb.
The semantics is given by \emph{LTL-structures} $\Wmf=(w_i)_{i\in \Zbb}$, where $w_i\subseteq P$.
We write
\begin{align*}
  \Wmf,i\models p &\text{ iff } p\in w_i \text{ if $p\in P$},
  &&&
  \Wmf,i\models\varphi\land\psi &\text{ iff } \Wmf,i\models\varphi\text{ and }\Wmf,i\models\psi,
  \\
  \Wmf,i\models\diamond_I \varphi &\text{ iff } \exists k\in I\colon \; \Wmf,i+k\models\varphi,
  &&&
  \Wmf,i\models\varphi\lor\psi &\text{ iff } \Wmf,i\models\varphi\text{ or }\Wmf,i\models\psi.
\end{align*}
In \dELHbm the following derived operators are used, where $n\ge 1$.
\begin{align}
  \diamondpm \varphi &:=\diamond_{(-\infty, \infty)}\varphi
  &&&
  \diamondp \varphi &:= \diamond_{[0,\infty)} \varphi
  \qquad\qquad\qquad
  \diamondm \varphi := \diamond_{(-\infty,0]} \varphi \nonumber
  \\
  \diamondc \varphi &:= \diamond_{(-\infty,0]}\varphi \land \diamond_{[0,\infty)}\varphi
  &&&
  \diamondc[n] \varphi &:= \bigvee\limits_{\substack{k,m\geq 0\\k+m=n-1}} (\diamond_{[-k,0]}\varphi \land \diamond_{[0,m]}\varphi) \label{eq:convex-diamond}
\end{align}
The operator~\diamondp is the \enquote{eventually} operator of classical LTL, and $\diamondm,\diamondpm$ are two variants that refer to the past, or to both past and future, respectively.
The operator~\diamondc requires that $\varphi$ holds \emph{both} in the past and in the future, thereby distinguishing time points that lie within an interval enclosed by time points at which~$\varphi$ holds.
This can be used to express the convex closure of time points, as described in the introduction. 
Finally, the operators~\diamondc[n] represent a metric variant of~\diamondc, requiring that different occurrences of~$\varphi$ are at most $n-1$ time points apart, \ie enclose an interval of length~$n$.

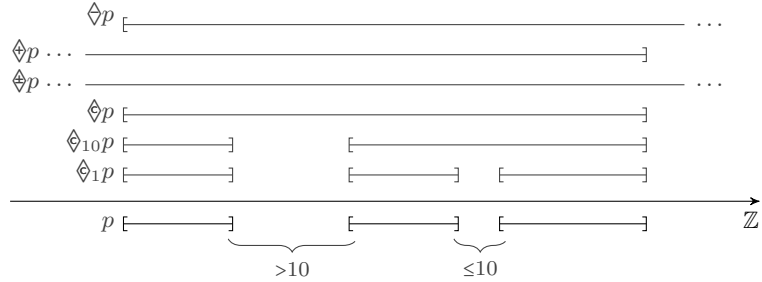
\begin{figure}
  \centering
\colorlet{alert}{darkgray}
\colorlet{alertt}{darkgray}
\newlength{\oph}
\setlength{\oph}{0.4cm}

\newlength{\opoff}
\setlength{\opoff}{0.3cm}

\tikzstyle{interval}=[]

\newcommand{\temporalbase}{
\draw[black,->] (0,0) -- (10,0);
\node[black] at (9.85,-0.25) {$\Zbb$};
\draw[interval, {[-]}] (1.5,-\opoff) -- (3  ,-\opoff);
\draw[interval, {[-]}] (4.5,-\opoff) -- (6  ,-\opoff);
\draw[interval, {[-]}] (6.5,-\opoff) -- (8.5,-\opoff);
\node[alert,anchor=east] at (1.5,-\opoff-0.1) {$p$};
}

\begin{tikzpicture}

  \temporalbase
  \begin{scope}[shift={(0,0\oph+\opoff)}]
    \draw[interval,{[-]},alertt] (1.5,0.05) -- (3,0.05);
    \draw[interval,{[-]},alertt] (4.5,0.05) -- (6,0.05);
    \draw[interval,{[-]},alertt] (6.5,0.05) -- (8.5,0.05);
    \node[alertt,anchor=east] at (1.5,0.1) {$\diamondc[1] p$};
  \end{scope}

  \begin{scope}[shift={(0,1\oph+\opoff)}]
    \draw[interval,{[-]},alertt] (1.5,0.05) -- (3,0.05);
    \draw[interval,{[-]},alertt] (4.5,0.05) -- (8.5,0.05);
    \node[alertt,anchor=east] at (1.5,0.1) {$\diamondc[10] p$};
  \end{scope}

  \begin{scope}[shift={(0,2\oph+\opoff)}]
    \draw[interval, {[-]},alertt] (1.5,0.05) -- (8.5,0.05);
    \node[alertt,anchor=east] at (1.5,0.1) {$\diamondc p$};
  \end{scope}

  \begin{scope}[shift={(0,3\oph+\opoff)}]
    \draw[interval, alertt, {-}] (1,0.05) -- (9,0.05);
    \node[alertt,anchor=east] at (0.5,0.1) {$\diamondpm p$};
    \node[alertt,anchor=west] at (9,0.05) {$\dots$};
    \node[alertt,anchor=east] at (1,0.05) {$\dots$};
  \end{scope}

  \begin{scope}[shift={(0,4\oph+\opoff)}]
    \draw[interval, alertt, {-]}] (1,0.05) -- (8.5,0.05);
    \node[alertt,anchor=east] at (0.5,0.1) {$\diamondp p$};
    \node[alertt,anchor=east] at (1,0.05) {$\dots$};
  \end{scope}

  \begin{scope}[shift={(0,5\oph+\opoff)}]
    \draw[interval, alertt, {[-}] (1.5,0.05) -- (9,0.05);
    \node[alertt,anchor=east] at (1.5,0.2) {$\diamondm p$};
    \node[alertt,anchor=west] at (9,0.05) {$\dots$};
  \end{scope}

  \draw[alert,-,decoration={brace,mirror,amplitude=5},decorate]
    (5.9,-0.5) -- (6.6,-0.5);
  \node[alertt] at (6.25,-0.9) {\footnotesize${\le}10$};

  \draw[alert,-,decoration={brace,mirror,amplitude=5},decorate]
    (2.9,-0.5) -- (4.6,-0.5);
  \node[alertt] at (3.75,-0.9) {\footnotesize${>}10$};
\end{tikzpicture}
\caption{The resulting intervals of the induced functions of the diamond operators when applied to the intervals at which $p$ holds (denoted below the timeline). While $\diamondc$ closes all gaps, $\diamondc_{10}$ does not close the first gap, since it is greater than $10$ time units.}
\label{fig:operators}
\end{figure}

To make the behavior of these operators more clear, we consider their semantics in a more abstract way:
given a set of time points where a certain information is available (\eg a diagnosis), described by a propositional variable~$p$, we consider the resulting set of time points at which $\da p$ holds, where \da is a placeholder for one of the operators defined above (we will similarly use $\db,\dc,\dd$ as placeholders for different $\diamond$-operators in the following).

  We consider the sets $\Dc := \{ \diamondc  \} \cup \{ \diamondc[i]\mid i \ge 1 \}$, $\Dpm = \{ \diamondm, \diamondp, \diamondpm\}$, and $\Ds :=  \Dpm \cup \Dc$ of \emph{diamond operators}.
  Each $\da\in\Ds$ induces a function $\da\colon 2^\Zbb\to 2^\Zbb$ with $\da(M):=\{i\mid \Wmf_M,i\models\da p \}$ for all $M\subseteq \Zbb$, with the LTL-structure $\Wmf_M:=(w_i)_{i\in \Zbb}$ such that $w_i := \{p\}$ if $i\in M$, and $w_i := \emptyset$ otherwise.

%
For an illustration see~\Cref{fig:operators}.
We will omit the parentheses in $\da(M)$ for a cleaner presentation.

If~$M$ is empty, then $\da M$ is also empty, for any $\da\in\Ds$.
For any non-empty $M\subseteq\Zbb$, the following expressions are obtained, where $\max M$ may be~$\infty$ and $\min M$ may be $-\infty$.
\begin{align*}
  \diamondpm M &= \Zbb
  &&&\!\!\!\!\!\!
  \diamondp M &= (-\infty,\max M]
  \quad
  \diamondm M = [\min M,\infty)
  \quad
  \diamondc M = [\min M, \max M]
  \\
  \diamondc[1]M &= M
  &&&\!\!\!\!\!\!
  \diamondc[n]M &= \{ i \in \Zbb \mid \exists j,k \in M \text{ with } j \leq i \leq k \text{ and } k-j < n \}
\end{align*}

Based on the operators in~\Ds we give the definition of \dELHbm in the following.

A \emph{\dELHbm concept} is built using the rule
$C ::= A \mid \top \mid \bot \mid C \sqcap C \mid \exists r.C \mid \da C$,
where $A \in \NC$, $\da\in\Ds$, and $r$ is a role.
Such a~$C$ is an \emph{\ELHb concept} (or \emph{atemporal concept}) if it does not contain any diamond operators.

A \emph{\dELHbm TBox} is a finite set of \emph{concept inclusions (CIs)} $C \sqsubseteq D$ and \emph{role inclusions (RIs)} $r \sqsubseteq s$, where $C$ is a \dELHb concept, $D$ is an atemporal concept, and $r,s \in \NR$ are roles.
Note that temporal roles are not allowed in \dELHbm.
An \emph{ABox} is a finite set of \emph{concept assertions} $A(a,i)$ and \emph{role assertions} $r(a,b,i)$, where $A\in\NC$, $r\in\NR$, $a,b \in \NI$, and $i\in\Zbb$.

The set of time points $i\in\Zbb$ occurring in $\Amc$ we denote as $\tem$. Also we assume each time point is encoded in binary.
In the following, we always assume a KB $\Kmc=\Tmc\cup\Amc$ to be given.

A \emph{temporal interpretation} $\Imf = ( \Delta^{\Imf}, (\Imc_i)_{i\in \Zbb})$, is a collection of interpretations $\Imc_i = (\Delta^\Imf, \cdot^{\Imc_i})$, $i \in \Zbb$, over~$\Delta^\Imf$.

The functions~$\cdot^{\Imc_i}$ are extended to temporal concepts as follows:
\begin{align*}
  (\da C)^{\Imc_i} &:= \big\{ d \in \Delta^\Imf \mid i \in \da\{ j \mid d \in C^{\Imc_j} \} \big\}
\end{align*}
$\Imf$ is a \emph{model} of (or \emph{satisfies})

  a concept inclusion $C\sqsubseteq D$ if $C^{\Imc_i}\subseteq D^{\Imc_i}$ holds for all $i\in\Zbb$,
  a role inclusion $r\sqsubseteq s$ if $r^{\Imc_i}\subseteq s^{\Imc_i}$ holds for all $i\in\Zbb$,
  a concept assertion $A(a,i)$ if $a\in A^{\Imc_i}$,
  a role assertion $r(a,b,i)$ if $(a,b)\in r^{\Imc_i}$,
  and the KB~\Kmc if it satisfies all axioms in~\Kmc.

This fact is denoted by $\Imf\models\alpha$, where $\alpha$ is an \emph{axiom} (\ie inclusion or assertion) or a KB.
Entailment of CIs and assertions in \dELHbm is \PTime-complete~\cite{BoFoKo-RRML}.

As in many lightweight temporal logics, diamonds are not allowed to occur on the right-hand side of CIs, because that would allow to simulate concept disjunction and make the logic intractable~\cite{AKL+-TIME07,GuJK-IJCAI16}.

As usual, we can simulate CIs involving complex concepts by introducing fresh concept and role names as abbreviations.
For example, $\exists r.\diamondm A\sqsubseteq B$ can be split into $\diamondm A\sqsubseteq A'$, and $\exists r.A'\sqsubseteq B$.
Hence, we can restrict ourselves w.l.o.g.\ to inclusions in the following \emph{normal form}:
\begin{equation}
  \da A\sqsubseteq B,\ A\sqsubseteq\exists r.B,\ A_1\sqcap A_2\sqsubseteq B,\ r\sqsubseteq s,\  \exists r.A\sqsubseteq B, \label{eq:normal1}
\end{equation}
where $\da \in \Ds$, $A,A_1,A_2,B\in \NC\cup\{\bot,\top\}$, and $r,s\in \NR$.

\paragraph{Conjunctive Query Answering.}

Let $\NV$ be a countably infinite set of \emph{variables}. The set of \emph{terms} is $\NT := \NV \cup \NI$.

A \emph{first-order query} $\phi(\mathbf{x})$ is a first-order formula built from \emph{concept atoms} $A(t)$ and \emph{role atoms} $r(t, t')$ with $A\in \NC$, $r\in \NR$, and $t_i \in \NT$, using the boolean connectives $(\land, \lor, \neg, \to)$ and universal and existential quantifiers $(\forall x, \exists x)$.
The free variables $\mathbf{x}$ of $\phi(\mathbf{x})$ are called \emph{answer variables} and we say that $\phi$ is $k$-ary if there are $k$ answer variables. The remaining variables are the \emph{quantified variables}. We use $\Var(\phi)$ to denote the set of all variables in $\phi$.
A query without any answer variables is called a \emph{Boolean query}.

Let $ \Imc  = (\Delta, \cdot^\Imc)$ be an interpretation. An \emph{assignment} $\pi \colon \Var(\phi) \to \Delta$

\emph{satisfies} $\phi$ in $ \Imc $, if $\Imc, \pi \models \phi$ under the standard semantics of first-order logic. We write $\Imc \models \phi$ if there is a satisfying assignment for~$\phi$ in~$\Imc$.
Let $\Kmc$ be a KB.
A $k$-tuple~$\mathbf{a}$ of individual names from $\Ind(\Kmc)$ is an \emph{answer} to~$\phi$ in~$\Imc$ if $\phi$ has a satisfying assignment~$\pi$ in~$\Imc$ with $\pi(\mathbf{x}) = \mathbf{a}$;
it is a \emph{certain answer} to~$q$ over~$\Kmc$ if it is an answer to~$q$ in all models of~$\Kmc$. 
We denote the set of all answers to $\phi$ in $\Imc$ by $\ans(\phi, \Imc)$, and the set of all certain answers to~$\phi$ over~$\Kmc$ by $\cert(\phi,\Kmc)$.

A \emph{conjunctive query} (CQ) $q(\mathbf{x})$ is a first-order query of the form $\exists \mathbf{y}.\, \varphi(\mathbf{x}, \mathbf{y})$, where $\varphi$ is a conjunction of atoms.
Abusing notation, we write $\alpha \in q$ if the atom~$\alpha$ occurs in~$q$, and conversely may treat a set of atoms as a conjunction.
The \emph{leaf variables}~$x$ in~$q$ are those that do not occur in any atoms of the form $r(x,y)$.
Clearly, $q$ is satisfied in an interpretation if there is a satisfying assignment for $\varphi(\mathbf{x}, \mathbf{y})$, which is often called a \emph{match} for $q$.
A CQ is \emph{rooted} if all variables are connected to an answer variable through (a sequence of) role atoms.

CQ answering over \ELHb KBs is \emph{combined first-order rewritable} \cite{lutz2009conjunctive}:
For any CQ $q$ and consistent KB $\Kmc=(\Tmc,\Amc)$ we can find a first-order query $q_\Tmc$ and a finite interpretation $\Imc_\Kmc^\mathrm{fin}$ such that
$\cert(q, \Kmc) = \ans(q_\Tmc, \Imc_\Kmc^\mathrm{fin})$.
Importantly, $\Imc_\Kmc^\mathrm{fin}$ is independent of~$q$, \ie can be reused to answer many different queries, while $q_\Tmc$ is independent of~$\Amc$, \ie each query can be rewritten without using the (possibly large) dataset.
The rewritability results are based crucially on the \emph{canonical model} property of \ELHb: For any consistent KB $\Kmc$ one can construct a model~$\Imc_\Kmc$ that is homomorphically contained in any other model. This is a very useful property since any match in the canonical model corresponds to matches in all other models of~$\Kmc$, and therefore $\cert(q, \Kmc) = \ans(q, \Imc_\Kmc)$ holds for all CQs~$q$.

For the complexity analysis, one considers the decision problem of whether a given tuple~$\mathbf{a}$ belongs to $\cert(q,\Kmc)$.
The complexity of query answering is usually viewed in two different ways: In \emph{combined complexity} the knowledge base, the data and the query are considered as input, while in \emph{data complexity}, which is closely related to rewritability, only the data is considered as input. The latter is often viewed as the more appropriate, since the queries and the knowledge base are tend to be fixed and relatively small, compared to potentially huge amounts of data, that are updated frequently.
CQ answering for \ELHb is \PTime-complete in data complexity and \NP-complete in combined complexity \cite{Rosa-DL07}.

\section{Conjunctive Queries With Negation}
\label{sec:ncqs}

We are interested in answering queries of the following form.

\begin{definition}
  \emph{Conjunctive queries with (guarded) negation} (NCQs) are constructed by extending CQs with \emph{negated concept atoms} $\neg A(t)$ and \emph{negated role atoms} $\neg r(t,t')$, such that, for any negated atom over terms $t$ (and $t'$) the query contains at least one positive atom over $t$ (and $t'$).
\end{definition}

An NCQ is \emph{rooted} if its variables are all connected via role atoms to an answer variable (from $\xbf$) or an individual name.
An NCQ is \emph{Boolean} if it does not have answer variables.
To determine whether $\Imc\models\phi$ holds for an NCQ~$\phi$ and an atemporal interpretation~\Imc, we use standard first-order semantics.

We first discuss different ways of handling the negated atoms, and then propose a new semantics that is based on a particular kind of \emph{minimal} canonical model.
For this, we consider an example based on real EHRs (ABoxes) from the MIMIC-III database~\cite{JPS+-SD16}, 
criteria (NCQs) from clinicaltrials.gov, and the large medical ontology SNOMED\,CT\footnote{https://www.snomed.org/snomed-ct} (the TBox). We omit here the \enquote{role groups} used in SNOMED\,CT, which do not affect the example. We also simplify the concept names and their definitions for ease of presentation.
We assume that the ABoxes have been extracted from EHRs by a natural language processing tool based, \eg on existing concept taggers like \cite{aronson2001effective,savova2010mayo}; of course, this extraction is an entire research field in itself, which we do not attempt to tackle in this paper.

\begin{example}
  \label{example:cancer-patient}
We consider three patients.
Patient~\BreastCancerPatient (patient 2693 in the MIMIC-III dataset) is diagnosed with breast cancer and an unspecified form of cancer (this often occurs when there are multiple mentions of cancer in a patient's EHR, which cannot be resolved to be the same entity).
Patient~\SkinAndBreastPatient (patient 32304 in the MIMIC-III dataset) suffers from breast cancer and skin cancer (\enquote{[S]tage IV breast cancer with mets to skin, bone, and liver}).
For~\SkinOfBreastPatient (patient 88432 in the MIMIC-III dataset), we know that~\SkinOfBreastPatient has breast cancer that involves the skin (\enquote{Skin, left breast, punch biopsy: Poorly differentiated carcinoma}).

Since SNOMED\,CT does not model patients, we add a special role name \emph{diagnosedWith} that connects patients with their diagnoses.
One can use this to express diagnoses in two ways.
First, one can explicitly introduce individual names for diagnoses in assertions like $\text{diagnosedWith}(\BreastCancerPatient,d_1)$, $\text{BreastCancer}(d_1)$, $\text{diagnosedWith}(\BreastCancerPatient,d_2)$, $\text{Cancer}(d_2)$, implying that these diagnoses are treated as distinct entities under the standard name assumption.
Alternatively, one can use complex assertions like $\exists\text{diagnosedWith}.\text{Cancer}(\BreastCancerPatient)$, which allows the logical semantics to resolve whether two diagnoses actually refer to the same object.
Since ABoxes only contain concept names, in this case one has to introduce auxiliary definitions like $\text{CancerPatient}\equiv\exists\text{diagnosedWith}.\text{Cancer}$ into the TBox.
We use both variants in our example, to illustrate their different behaviors.

  We obtain the KB $\Kmc_C$, containing knowledge about different kinds of cancers and cancer patients, together with information about the three patients. The information about cancers is taken from SNOMED\,CT (in simplified form):
  \begin{align*}
  \text{SkinCancer} &\equiv \text{Cancer} \sqcap \exists\text{findingSite}.\text{SkinStructure}                 \\
  \text{BreastCancer} &\equiv \text{Cancer} \sqcap \exists\text{findingSite}.\text{BreastStructure}             \\
  \text{SkinOfBreastCancer} &\equiv \text{Cancer} \sqcap \exists\text{findingSite}.\text{SkinOfBreastStructure}           \\
  \text{SkinOfBreastStructure} &\sqsubseteq \text{BreastStructure}\sqcap \text{SkinStructure}
  \end{align*}
  The EHRs are compiled into several assertions per patient:
  \begin{align*}
    \text{Patient~\BreastCancerPatient:} &\ \text{BreastCancerPatient}(\BreastCancerPatient),\  \text{CancerPatient}(\BreastCancerPatient) \\
    \text{Patient~\SkinAndBreastPatient:} &\ \text{SkinCancerPatient}(\SkinAndBreastPatient),\  \text{BreastCancerPatient}(\SkinAndBreastPatient) \\
    \text{Patient~\SkinOfBreastPatient:} &\ \text{diagnosedWith}(\SkinOfBreastPatient, c_\SkinOfBreastPatientIdx),\ \text{SkinOfBreastCancer}(c_\SkinOfBreastPatientIdx)
  \end{align*}
  Additionally, we add the following auxiliary definitions to the TBox:
  \begin{align*}
    \text{CancerPatient} &\equiv \exists\text{diagnosedWith}.\text{Cancer}                                                          \\
    \text{SkinCancerPatient} &\equiv \exists\text{diagnosedWith}.\text{SkinCancer}                                                  \\
    \text{BreastCancerPatient} &\equiv \exists\text{diagnosedWith}.\text{BreastCancer}
  \end{align*}
  For example, skin cancers and breast cancers are cancers occurring at specific parts of the body (\enquote{body structure} in SNOMED\,CT), and a breast cancer patient is someone who is diagnosed with breast cancer. This means that, in every model of~$\Kmc_C$, every object that satisfies BreastCancerPatient (in particular~\SkinAndBreastPatient) must have a diagnosedWith-connected object that satisfies BreastCancer, and so on.
  Moreover, a cancer may also occur in the skin of the breast, \ie a part of the body that is classified as both a \enquote{skin structure} and as a \enquote{breast structure}.

  For a clinical trial,\footnote{https://clinicaltrials.gov/ct2/show/NCT01960803} we want to find patients that have
  \enquote{breast cancer}, but not \enquote{breast cancer that involves the skin.}

  This can be translated into an NCQ:
  \begin{align*}
    q_B(x) := \exists y,z.\, &\text{diagnosedWith}(x,y) \land \text{Cancer}(y) \land \text{findingSite}(y,z) \land{} \\
    & \text{BreastStructure}(z) \land \neg \text{SkinStructure}(z)
  \end{align*}
We know that~\BreastCancerPatient is diagnosed with BreastCancer as well as Cancer.

Since the former is more specific, we assume that the latter refers to the same BreastCancer.

However, since we have no information about an involvement of the skin, \BreastCancerPatient should be returned as an answer to~$q_B$.

We know that~\SkinAndBreastPatient suffers from cancer in the skin and the breast, but not if the skin of the breast is also affected.
Since neither location is implied by the other, we assume that they refer to distinct areas. \SkinAndBreastPatient should thus be an answer to~$q_B$.

In the case of~\SkinOfBreastPatient, it is explicitly stated that it is the same cancer that is occurring (not necessarily exclusively) at the skin of the breast.
In this case, the ABox assertions override the distinctness assumption we made for~\SkinAndBreastPatient.
Thus, \SkinOfBreastPatient should not be an answer to~$q_B$.
\hfill$\blacksquare$
\end{example}

\begin{figure}[tb]
  \resizebox{\linewidth}{!}{\input{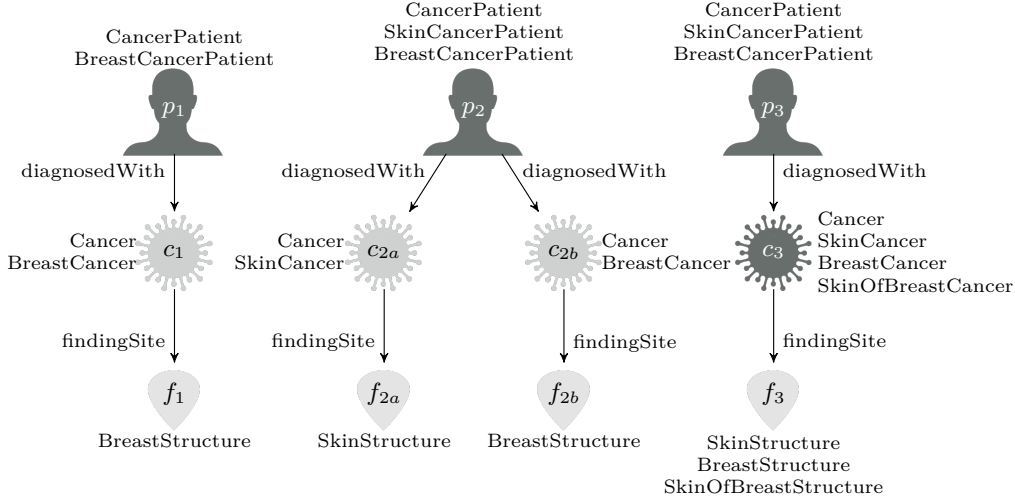}}
  \caption{The minimal canonical model~$\Imc_{\Kmc_C}$.
  The named individuals $p_1,p_2, p_3$ and $c_3$ are depicted by dark grey figures, the remaining anonymous objects by light grey. }
  \label{fig:canonical-model}
\end{figure}

In practice, more complicated cases than in our example can occur:
The nesting of anonymous objects will be deeper and more branched when using large biomedical ontologies.
For example, in SNOMED\,CT it is possible to describe many details of a cancer, such as the kind of cancer, whether it is a primary or secondary cancer, and in which part of the body it is found. This means that even a single assertion can lead to the introduction of multiple levels of anonymous objects in the canonical model.
In some ontologies there are even cyclic concept inclusions, which lead to infinitely many anonymous individuals, \eg in the GALEN ontology\footnote{http://www.opengalen.org/}.
We focus on \Cref{example:cancer-patient} in this paper, to illustrate the relevant issues in a clear and easy to follow manner.

We now evaluate existing semantics on this example.

\paragraph{Standard Certain Answer Semantics} as defined in \Cref{sec:preliminaries} is clearly not suited here, because one can easily construct a model of~$\Kmc_C$ in which~$ c_\BreastCancerPatientIdx$ is also a skin cancer, and hence~\BreastCancerPatient is not an element of $\cert(q_B,\Kmc_C)$.
Moreover, under certain answer semantics answering CQs with guarded negation is already \textsc{coNP}-complete \cite{GIKK-JWS15}, and hence not (combined) rewritable.

\paragraph{Epistemic Logic} allows us to selectively apply closed-world reasoning using the modal knowledge operator~$\mathbf{K}$.
For a formula $\mathbf{K}\varphi$ to be true, it has to hold in all \enquote{connected worlds}, which is often considered to mean all possible models of the KB, adopting an $S5$-like view \cite{calvanese2006epistemic}.

For~$q_B$, we could read $\neg \text{SkinStructure}(z)$ as \enquote{not known to be a skin structure}, \ie $\neg\mathbf{K}\text{SkinStructure}(z)$.
Consider the model $\Imc_{\Kmc_C}$ in \Cref{fig:canonical-model} and the assignment $\pi = \{ x \mapsto \SkinOfBreastPatient, y \mapsto c_\SkinOfBreastPatientIdx, z \mapsto f_\SkinOfBreastPatientIdx \}$, for which we want to check whether it is a match for~$q_B$. Under epistemic semantics, $\neg\mathbf{K}\text{SkinStructure}(z)$ is considered true if \Kmc has a (different) model in which~$f_\SkinOfBreastPatientIdx$ does not belong to~$\text{SkinStructure}$.
However, $f_\SkinOfBreastPatientIdx$ is an anonymous object, and hence its name is not fixed.
For example, we can easily obtain another model by renaming~$f_\SkinOfBreastPatientIdx$ to~$f_\BreastCancerPatientIdx$ and vice versa.
Then~$f_\SkinOfBreastPatientIdx$ would not be a skin structure, which means that $\neg\mathbf{K}\text{SkinStructure}(z)$ is true in the original model~$\Imc_{\Kmc_C}$, which is not what we expected.
This is a known problem with epistemic first-order logics \cite{wolter2000first}.

\paragraph{Skolemization} can enforce a stricter comparison of anonymous objects between models. The inclusion
$\text{SkinOfBreastCancer} \sqsubseteq \exists\text{findingSite}.\text{SkinOfBreast}$
could be rewritten as the first-order sentence
\[ \forall x.\,\Big(\text{SkinOfBreastCancer}(x) \to \text{findingSite}\big(x,f(x)\big) \land \text{SkinOfBreast}\big(f(x)\big)\Big), \]
where $f$ is a fresh function symbol.
This means that~$c_\SkinOfBreastPatientIdx$ would be connected to a finding site that has the unique name $f(c_\SkinOfBreastPatientIdx)$ in every model.
Queries would be evaluated over Herbrand models only.
Hence, for evaluating $\neg\mathbf{K}\text{SkinStructure}(z)$ when $z$ is mapped to $f(c_\SkinOfBreastPatientIdx)$, we would only be allowed to compare the behavior of $f(c_\SkinOfBreastPatientIdx)$ in other Herbrand models.
The general behavior of this anonymous individual is fixed, however, since in all Herbrand models it is \emph{the} finding site of~$c_\SkinOfBreastPatientIdx$.
While this improves the comparison by introducing pseudo-names for all anonymous individuals, it limits us in different ways:
Since \SkinOfBreastPatient is inferred to be a BreastCancerPatient, the Skolemized version of $\text{BreastCancerPatient} \sqsubseteq \exists\text{diagnosedWith}.\text{BreastCancer}$ introduces a new successor~$g(\SkinOfBreastPatient)$ of~\SkinOfBreastPatient satisfying BreastCancer, which, together with the definition of BreastCancer, means that \SkinOfBreastPatient is an answer to~$q_B$ since there is an additional breast cancer diagnosis that does not involve the skin.

\paragraph{Datalog-based Ontology Languages} with negation \cite{HKLG-PODS13,ArGP-PODS14} are closely related to Skolemized ontologies, since their semantics is often based on the so-called \emph{Skolem chase} \cite{Marn-PODS09}. This is closer to the semantics we propose in Section~\ref{sec:semantics}, in that a single canonical model is used for all inferences. However, it suffers from the same drawback of Skolemization described above, due to superfluous successors. To avoid this, our semantics uses a special minimal canonical model (see Definition~\ref{def:canonical-model}), which is similar to the \emph{restricted chase} \cite{FKMP-TCS05} or the \emph{core chase} \cite{DeNR-PODS08}, but always produces a unique model without having to merge domain elements. Our minimal canonical model can also be seen as an instance of the approach in \cite{K2020}, which identifies conditions under which a chase variant produces a unique core model (for more details on this connection see Section~3.5.3 in \cite{Fork-20}). However, \cite{K2020} focuses on cases where this model is finite, which is not the case for general \ELHb ontologies. To the best of our knowledge, there exist no complexity results for Datalog-based languages with negation over these other chase variants.

\paragraph{Closed Predicates} are a way to declare, for example, the concept name SkinStructure as \enquote{closed}, which means that all skin structures must be declared explicitly, and no other SkinStructure object can exist \cite{LuSW-IJCAI13,AhOS-IJCAI16}.
This provides a way to give answers to negated atoms as in~$q_B$.
However, as explained in the introduction, this mechanism is not suitable for anonymous objects since it means that only named individuals can satisfy SkinStructure.
When applied to~$\Kmc_C$, the result is even worse: Since there is no (named) SkinStructure object, no skin structures can exist at all and~$\Kmc_C$ becomes inconsistent.
Closed predicates are appropriate in cases where the KB contains a full list of all instances of a certain concept name, and no other objects should satisfy it;
but they are not suitable to infer negative information about anonymous objects.
Moreover, CQ answering with closed predicates in \ELHb is already \textsc{coNP}-hard \cite{LuSW-IJCAI13}.

\paragraph{Summary} All of this should not be read as saying that these semantics are bad, just that they have not been developed with our use case involving anonymous objects in mind. Only Skolemization deals with anoynmous objects by giving them unique names, but this forces all of them to be distinct. Our new semantics remedies this, but at the same time tries to stay as close as possible to the behavior of the existing semantics when restricted to the known objects in the ABox.

\subsection{Semantics for NCQs}
\label{sec:semantics}

We propose to answer NCQs over a special canonical model of the knowledge base.
On the one hand, this eliminates the problem of tracking anonymous objects across different models, and on the other hand enables us to encode our assumptions directly into the construction of the model.
In particular, we should only introduce the minimum necessary number of anonymous objects since, unlike in standard CQ answering, the precise shape and number of anonymous objects has an impact on the semantics of negated atoms.

Given~$\Kmc_C$, in contrast to the Skolemized semantics, we will not create both a generic \enquote{Cancer} and another \enquote{BreastCancer} successor for~\BreastCancerPatient, because the BreastCancer is also a Cancer, and hence the first object is redundant.
Therefore, in the minimal canonical model of~$\Kmc_C$ depicted in \Cref{fig:canonical-model}, for patient \BreastCancerPatient only one successor is introduced to satisfy the definitions of both BreastCancerPatient and CancerPatient at the same time.
In contrast, $\SkinAndBreastPatient$ has two successors, because BreastCancer and SkinCancer do not imply each other.
Finally, for~\SkinOfBreastPatient the ABox contains a single successor that is a SkinOfBreastCancer, which implies a single findingSite-successor that satisfies both SkinStructure and BreastStructure.

To detect whether an object required by an existential restriction $\exists r.A$ is redundant, we use the following notion of minimality.
\begin{definition}[Structural Subsumption]
  Let $\exists r.A$, $\exists t.B$ be concepts with $A,B \in \NC$ and $r, t \in \NR$. We say that $\exists r.A$ is \emph{structurally subsumed} by~$\exists t.B$ (written $\exists r.A \sqsubseteq_\Tmc^s \exists t.B$) if $r \sqsubseteq_\Tmc t$ and $A\sqsubseteq_\Tmc B$.

  Given a set~$V$ of existential restrictions, we say that $\exists r.A \in V$ is \emph{minimal w.r.t.~$\sqsubseteq_\Tmc^s$ (in~$V$)} if there is no $\exists t.B \in V$ such that $\exists t.B \sqsubseteq_\Tmc^s \exists r.A$.

  A CQ $q_1(\mathbf{x})$ is \emph{structurally subsumed} by a CQ $q_2(\mathbf{x})$ with the same answer variables (written $q_1 \sqsubseteq_\Tmc^s q_2$) if, for all $x,y \in \mathbf{x}$, it holds that
  $$
    \bigsqcap\limits_{\alpha(x) \in q_1} \alpha \sqsubseteq_\Tmc \bigsqcap\limits_{\alpha(x) \in q_2} \alpha, \text{ and }
    \bigsqcap\limits_{\alpha(x,y) \in q_1} \alpha \sqsubseteq_\Tmc \bigsqcap\limits_{\alpha(x,y) \in q_2} \alpha,
  $$
  where role conjunction is interpreted in the standard way~\cite{BCM+-07}.
\end{definition}

In contrast to standard subsumption, $\exists r.A$ is not structurally subsumed by~$\exists t.B$ \wrt the TBox $\Tmc=\{\exists r.A\sqsubseteq\exists t.B\}$, as neither $r\sqsubseteq_\Tmc t$ nor $A\sqsubseteq_\Tmc B$ hold. Similarly, structural subsumption for CQs considers all (pairs of) variables separately.

We use this notion to define the minimal canonical model.

\begin{definition}[Minimal Canonical Model]
  \label{def:canonical-model}
  Let $\Kmc = (\Tmc, \Amc)$ be an \ELHb KB.
  We construct the \emph{minimal canonical model} $\Imc_\Kmc$ of $\Kmc$ as follows:
  \begin{enumerate}
    \item Set $\Delta^{\Imc_{\Kmc}} := \NI$ and $a^{\Imc_{\Kmc}} := a$ for all $a \in \NI$.
    \item Define $A^{\Imc_{\Kmc}} := \{ a \mid \Kmc \models A(a) \}$ for all $A \in \NC$ and
    $r^{\Imc_{\Kmc}} := \{ (a,b) \mid \Kmc \models r(a,b) \}$ for all $r \in \NR$.
    \item Repeat:
    \begin{enumerate}
      \item Select an element $d \in \Delta^{\Imc_\Kmc}$ that has not been selected before and let\\
      $V := \{ \exists r.B \mid$ $d \in A^{\Imc_{\Kmc}}$ and $d \not\in (\exists r. B)^{\Imc_\Kmc}$ with $A \sqsubseteq_\Tmc \exists r.B$,
      $A,B\in \NC\}$.
      \item For each $\exists r.B \in V$ that is minimal \wrt $\sqsubseteq_\Tmc^s$, add a fresh element~$e$ to~$\Delta^{\Imc_\Kmc}$, for each $B \sqsubseteq_\Tmc A$ add $e$ to $A^{\Imc_\Kmc}$, and for each $r \sqsubseteq_\Tmc s$ add $(d,e)$ to $s^{\Imc_\Kmc}$.
    \end{enumerate}
  \end{enumerate}
  By $\Imc_\Amc$ we denote the restriction of~$\Imc_\Kmc$ to named individuals, \ie the result of applying only Steps~1 and~2, but not Step~3.
\end{definition}

If Step~3 is applied fairly, \ie such that each new domain element that is created in~(b) is eventually also selected in~(a), then $\Imc_\Kmc$ is indeed a model of~$\Kmc$ (if \Kmc is consistent at all).
In particular, all required existential restrictions are satisfied at each domain element, because the existential restrictions that are minimal w.r.t.~$\sqsubseteq_\Tmc^s$ entail all others.

Moreover, $\Imc_\Kmc$ satisfies the properties expected of a canonical model \cite{lutz2009conjunctive,EOS+-AAAI12}: It can be homomorphically embedded into any other model of~\Kmc, and therefore $\cert(q,\Kmc) = \ans(q,\Imc_\Kmc)$ holds for all CQs~$q$.

It turns out that the minimal canonical model is also a \emph{core}: Each endomorphism of $\Imc_\Kmc$, \ie a homomorphism from $\Imc_\Kmc$ into inself, is injective, surjective and preserves negation (see also \cite{Fork-20}).

\begin{lemma}\label{lem:core}
  Every endomorphism of~$\Imc_\Kmc$ is a strong isomorphism.
\end{lemma}
\begin{proof*}
  To show this, we prove the stronger statement that the only endomorphism on~$\Imc_\Kmc$ is the identity.
  Let $\Imc_0$, $\Imc_1$, \dots be the interpretations obtained in the construction of~$\Imc_\Kmc$ before each application of Step~3.
  We show by induction on~$i$ that there are $h_0$, $h_1$, \dots such that $h_i$ is the only possible homomorphism from~$\Imc_i$ to~$\Imc_\Kmc$ and $h_i$ and~$h_{i+1}$ agree on $\Delta^{\Imc_i}$, that is $h_{i+1}(d) = h_i(d)$ for all $d \in \Delta^{\Imc_i}$. The endomorphism is then obtained in the limit as $h = \bigcup_{i\ge 0}h_i$.
  By definition, the homomorphism~$h_0$ has to map all constants onto itself, \ie $h_0(a) = a$ for all $a \in \NI$, and is therefore the identity function.
  For the induction step, assume that $h_i$ has already been defined.
  To define $h_{i+1}$, assume that $d \in \Delta^{\Imc_i}$ was picked in Step 3(a) and $V$ is the set as defined in Definition~\ref{def:canonical-model}.
  For each $\exists r.B \in V$ that is minimal \wrt structural subsumption, in Step 3(b) exactly one successor~$e$ is introduced.
  Suppose there would be another successor~$e'$ of~$d$, introduced through some minimal $\exists r'.B'$ and~$e$ could be mapped to~$e'$.
  This would imply that $\exists r'.B' \sqsubseteq_\Tmc^s \exists r.B$, which is a contradiction since we assumed $\exists r.B$ to be minimal.
  Hence, such an $e'$ cannot exist and therefore the only possibility is to map $e$ onto itself.
  Therefore, the only possibility is to define $h_{i+1} := h_i \cup \{e \mapsto e\}$, which is the identity function.
  \qed
\end{proof*}

We can now adopt a result from graph theory and cores, namely Theorem 11 in \cite{bauslaugh1995core}, which shows that, if two structures are homomorphically equivalent, then their cores are isomorphic.
Since all canonical models are homomorphically equivalent by definition, this implies that every consistent \ELHb-KB has a unique core model (up to isomorphism), which is why a minimal canonical model is also \emph{the} minimal canonical model of $\Kmc$.

We now define the semantics of NCQs as described before, \ie by evaluating them as first-order formulas over the minimal canonical model~$\Imc_\Kmc$, which ensures that our semantics is compatible with the usual certain-answer semantics for CQs.

\begin{definition}[Minimal-World Semantics]
  Let \Kmc be a consistent \ELHb KB.
  The \emph{(minimal-world) answers} to an NCQ~$q$ over \Kmc are $\mwa(q,\Kmc):=\ans(q,\Imc_\Kmc)$.
\end{definition}

For \Cref{example:cancer-patient}, we get $\mwa(q_B, \Kmc_C) = \{\BreastCancerPatient, \SkinAndBreastPatient\}$ (see \Cref{fig:canonical-model}), which is exactly as intended.
Unfortunately, in general the minimal canonical model is infinite, and we cannot evaluate the answers directly. Hence, we employ a rewriting approach to reduce NCQ answering over the minimal canonical model to (first-order) query answering over~$\Imc_\Amc$ only.

\section{A Combined Rewriting for NCQs}\label{sec:atemp-rewriting}

We show that NCQ answering is combined first-order rewritable.

As target representation, we obtain first-order queries of a special form.

\begin{definition}[Filtered query]
  \label{def:filterd-fo-query}
  Let $\Kmc = (\Tmc, \Amc)$ be an \ELHb KB.
  A \emph{filter} on a variable~$z$
  is a first-order expression $\psi(z)$ of the form
  \begin{equation}
    \big(\exists z'. \psi^+(z,z')\big) \to \big(\exists z'. \psi^+(z,z') \land \psi^-(z,z') \land \Psi\big)
  \end{equation}
  where $\psi^+(z,z')$ is a conjunction of atoms of the form $A(z')$ or $r(z,z')$, that contains at least one role atom,
  $\psi^-(z,z')$ is a conjunction of negated atoms $\neg A(z')$ or $\neg r(z,z')$, and $\Psi$ is a (possibly empty) set of filters on $z'$.


  A \emph{filtered query} $\phi$ is of the form $\exists\mathbf{y}. \big(\varphi(\mathbf{x}, \mathbf{y}) \land \Psi\big)$
  where $\exists\mathbf{y}.\varphi(\mathbf{x}, \mathbf{y})$ is an NCQ and $\Psi$ is a set of filters on leaf variables in $\varphi$.
  It is \emph{rooted} if $\exists\mathbf{y}.\varphi(\mathbf{x}, \mathbf{y})$ is rooted.
\end{definition}
Note that every NCQ is a filtered query where the set of filters $\Psi$ is empty.

We will use filters to check for the existence of \enquote{typical} successors, \ie role successors that behave like the ones that are introduced by the canonical model construction to satisfy an existential restriction.
In particular, a typical successor does not satisfy any superfluous concept or role atoms.
For example, in \Cref{fig:canonical-model} the element~$c_\BreastCancerPatientIdx$ introduced to satisfy $\exists\text{diagnosedWith}.\text{BreastCancer}$ for~\BreastCancerPatient is a typical successor, because it satisfies only BreastCancer and Cancer and not, \eg SkinCancer.
In contrast, the diagnosedWith-successor~$c_\SkinOfBreastPatientIdx$ of~$\SkinOfBreastPatient$ is atypical, since the ontology does not contain an existential restriction $\exists\text{diagnosedWith}.\text{SkinOfBreastCancer}$ that could have introduced such a successor in the canonical model.

The idea of the rewriting procedure is to not only rewrite the positive part of the query, as in \cite{EOS+-AAAI12,bienvenu2015ontology}, but to also ensure that no critical information is lost. This is accomplished by rewriting the negative parts and by saving the structure of the eliminated part of the query in the filter. A filter on~$z$ ensures that the rewritten query can only be satisfied by mapping~$z$ to an anonymous individual in the canonical model, or to a named individual that behaves in a similar way.

\begin{definition}[Rewriting]
  \label{def:ecq-rewriting}
  Let $\Kmc = (\Tmc, \Amc)$ be a KB and
  $\phi = \exists\mathbf{y}. \varphi(\mathbf{x}, \mathbf{y}) \land \Psi$ be a filtered query.
  We write $\phi \to_\Tmc \phi'$ if $\phi'$ can be obtained from $\phi$ by applying the following steps:

  \begin{enumerate}[label=(S\arabic*)]
    \item Select a quantified leaf variable~$\hat{x}$ in~$\varphi$.
    Let $\hat{y}$ be a fresh variable and select
    \begin{align*}
      \PRED &:= \{ y \mid r(y,\hat{x}) \in \varphi \} \cup \{ y \mid \neg r(y,\hat{x}) \in \varphi \} & \text{(predecessors of~$\hat{x}$),} \\
      \POS &:= \{ A(\hat{x}) \in \varphi \} \cup \{ r(\hat{y}, \hat{x}) \mid r(y, \hat{x}) \in \varphi \} & \text{(positive atoms for~$\hat{x}$),} \\
      \NEG &:= \{ \neg A(\hat{x}) \in \varphi \} \cup \{ \neg r(\hat{y}, \hat{x}) \mid \neg r(y,\hat{x}) \in \varphi \} & \text{(negative atoms for~$\hat{x}$).}
    \end{align*}
    \item Select some $M \sqsubseteq_\Tmc \exists s.N$ with $M,N \in \NC$ that satisfies all of the following:
    \begin{enumerate}[(a)]
      \item $s(\hat{y}, \hat{x}) \land N(\hat{x}) \sqsubseteq_\Tmc^s \POS$, and
      \item $s(\hat{y}, \hat{x}) \land N(\hat{x}) \not\sqsubseteq_\Tmc^s \alpha$ for all $\neg \alpha \in \NEG$.
    \end{enumerate}
    \item Let $\mathcal{M}'$ be the set of all $M' \in \NC$ such that $ M' \sqsubseteq_\Tmc \exists s'.N'$ with $N'\in \NC$,
    \begin{enumerate}[(a)]
      \item $\exists s'.N' \sqsubseteq_\Tmc^s \exists s.N$ (where $\exists s.N$ was chosen in (S2)), and
      \item $s'(\hat{y}, \hat{x}) \land N'(\hat{x}) \sqsubseteq_\Tmc^s \alpha$ for some $\neg\alpha \in \NEG$.
    \end{enumerate}
    \item Drop from $\varphi$ every atom that contains $\hat{x}$.
    \item Replace all variables $y \in \PRED$ in $\varphi$ with $\hat{y}$.
    \item Add the atoms $M(\hat{y})$ and $\{\neg M'(\hat{y}) \mid M' \in \mathcal{M}' \}$ to $\varphi$.
    \item
    Set the new filters to $\Psi' := \Psi \cup \{ \psi^*(\hat{y}) \} \setminus \Psi_{\hat{x}} $, where
    $\Psi_{\hat{x}} := \{ \psi(\hat{x}) \in \Psi \}$ and
    \begin{align*}
      \psi^*&(\hat{y}) := \big(\exists \hat{x}.\,s(\hat{y}, \hat{x}) \land N(\hat{x})\big)\to \big(\exists \hat{x}.\,s(\hat{y}, \hat{x}) \land N(\hat{x}) \land \NEG \land \Psi_{\hat{x}}\big).
    \end{align*}

  \end{enumerate}
  We write $\phi \to_\Tmc^* \phi'$ if there exists a finite sequence $\phi \to_\Tmc \dots \to_\Tmc \phi'$.
  Furthermore, let $\rewo_\Tmc(\phi) := \{ \phi' \mid \phi \to_\Tmc^* \phi' \}$ denote the set of all rewritings of $\phi$.
\end{definition}

Note that $\rewo_\Tmc(\phi)$ may be infinite.
However, for rooted NCQs it is finite and we show later that even for non-rooted NCQs it suffices to consider a finite subset of $\rewo_\Tmc(\phi)$ (see Lemma~\ref{lemma:non-rooted-bound}).
To see the former claim, observe that there is only a finite number of possible subsumptions $M \sqsubseteq_\Tmc \exists s.N$ that can be used for rewriting steps.
Additionally, in every step one variable ($\hat{x}$) is eliminated from the NCQ part of the filtered query.
If the query is rooted, there always exists at least one predecessor that is renamed to $\hat{y}$, hence the introduction of $\hat{y}$ never increases the number of variables.
Finally, it is easy to see that rewriting a rooted query always yields a rooted query.

The rewriting of~$\NEG$ to the new negated atoms (via~$\Mmc'$ in (S6)) ensures that we do not lose important exclusion criteria, which may result in too many answers.
Similarly, the filters exclude atypical successors in the ABox that may result in spurious answers.
Both of these constructions are necessary.

\begin{example}
  Consider the query~$q_B$ from \Cref{example:cancer-patient}. Using \Cref{def:ecq-rewriting}, we obtain the first-order queries $\phi_B = q_B$, $\phi_B'$, and $\phi_B''$, where
  \begin{align*}
    \phi_B' &= \exists y.\, \text{diagnosedWith}(x,y) \land \text{BreastCancer}(y) \land \neg \text{SkinOfBreastCancer}(y) \land{} \\
    & \Big( \big(\exists z.\, \text{findingSite}(y,z) \land \text{BreastStructure}(z)\big) \to \\
    &\;\;\big(\exists z.\, \text{findingSite}(y,z) \land \text{BreastStructure}(z) \land \neg \text{SkinStructure}(z)\big) \Big) \\
\intertext{
    results from choosing~$z$ in~(S1), $\text{BreastCancer}\sqsubseteq_{\Kmc_C}\exists\text{findingSite}.\text{BreastStructure}$ in~(S2), and computing $\Mmc'=\{\text{SkinOfBreastCancer}\}$ in~(S3), and
}
    \phi_B'' &= \text{BreastCancerPatient}(x) \land{} \\
    & \Big( (\exists y.\, \text{diagnosedWith}(x,y) \land \text{BreastCancer}(y)) \to \\
    &\;\;(\exists y.\, \text{diagnosedWith}(x,y) \land \text{BreastCancer}(y) \land \neg \text{SkinOfBreastCancer}(y)) \land{} \\
    &\;\;\;\; \big( (\exists z.\, \text{findingSite}(y,z) \land \text{BreastStructure}(z)) \to \\
    &\;\;\;\;\;\;(\exists z.\, \text{findingSite}(y,z) \land \text{BreastStructure}(z) \land \neg \text{SkinStructure}(z)) \big)
    \Big)
  \end{align*}
  is obtained due to $\text{BreastCancerPatient}\sqsubseteq_{\Kmc_C}\exists\text{diagnosedWith}.\text{BreastCancer}$.
  We omitted the redundant atoms $\text{Cancer}(y)$ for clarity.

  The finite interpretation~$\Imc_{\Amc_C}$ can be seen in \Cref{fig:canonical-model} by ignoring all light grey figures.
  When computing the answers over $\Imc_{\Amc_C}$, we obtain
  \begin{align*}
    \ans(\phi_B, \Imc_{\Amc_C}) &= \emptyset,\
    \ans(\phi'_B, \Imc_{\Amc_C}) = \emptyset, \text{ and }
    \ans(\phi''_B, \Imc_{\Amc_C}) = \{\BreastCancerPatient,\SkinAndBreastPatient\}.
  \end{align*}

  For~$\phi'_B$, the conjunct $\lnot\text{SkinOfBreastCancer}(y)$ is necessary to exclude \SkinOfBreastPatient as an answer.
%
  In~$\phi''_B$, \SkinOfBreastPatient is excluded due to the filter that detects~$c_\SkinOfBreastPatientIdx$ as an atypical successor, because it satisfies not only BreastCancer, but also SkinOfBreastCancer.
  Hence, both~(S6) and~(S7) are necessary steps in our rewriting.
  \hfill$\blacksquare$
\end{example}

\subsection{Correctness}

In \Cref{def:ecq-rewriting}, the new filter $\psi^*(\hat{y})$ may end up inside another filter expression after applying subsequent rewriting steps, \ie by rewriting w.r.t.~$\hat{y}$.

In this case, however, the original structure of the rewriting is preserved, including all internal filters as well as the atoms $M(\hat{y})$, which are included implicitly by $\exists s.N \sqsubseteq M$, and  $\{\neg M'(\hat{y}) \mid M' \in \mathcal{M}' \}$, which are included in $\NEG$.
We exploit this behavior to show that, whenever a rewritten query is satisfied in the finite interpretation~$\Imc_\Amc$, then it is also satisfied in~$\Imc_\Kmc$.
This is the most interesting part of the correctness proof, because it differs from the known constructions for ordinary CQs, for which this step is trivial.

\begin{lemma}
  \label{lemma:filtered-fo-query-data-2}
  Let $\Kmc = (\Tmc, \Amc)$ be a consistent \ELHb KB and $\phi$ be an NCQ.
  Then, for all $\phi' \in \rewo_\Tmc(\phi)$,
  \[
  \ans(\phi', \Imc_\Amc) \subseteq \mwa(\phi', \Kmc).
  \]
\end{lemma}

\begin{proof*}
Let $\phi'=\exists\mathbf{y}.(\varphi(\mathbf{x},\mathbf{y})\land\Psi)$ and $\pi$ be an assignment of $\mathbf{x},\mathbf{y}$ to $\NI$ such that $\Imc_\Amc, \pi \models \varphi(\mathbf{x}, \mathbf{y})$.
Since $\Imc_\Amc$ and $\Imc_\Kmc$ coincide on the domain~$\NI$, we also have $\Imc_\Kmc,\pi\models\varphi(\mathbf{x},\mathbf{y})$.

Consider any filter $\psi(z) = \exists z'. \psi^+(z,z') \to \exists z'.(\beta(z,z')\land \Psi^*)$ in $\Psi$, where
$\beta(z,z') := \psi^+(z,z') \land \psi^-(z,z') $. Then $\psi(z)$ was introduced at some point during the rewriting,
suppose by selecting

$M \sqsubseteq_\Tmc \exists s. N$ in~(S2).
This means that $\varphi$ contains the atom $M(z)$, and hence $d:=\pi(z)$ is a named individual that is contained in~$M^{\Imc_\Amc}\subseteq M^{\Imc_\Kmc}$.
By~(S2), this means that $\Imc_\Kmc, \pi \models \exists z'.\psi^+(z,z')$, and we have to show that $\Imc_\Kmc, \pi \models \exists z'.(\beta(z,z')\land \Psi^*)$:

\begin{enumerate}
  \item If $\Imc_\Amc,\pi\models\exists z'.\beta(z,z')$, then
  $\Imc_\Kmc,\pi\models\exists z'.\beta(z,z')$ by the same argument as for $\varphi(\mathbf{x}, \mathbf{y})$ above, and we can proceed by induction on the structure of the filters to show that the inner filters~$\Psi^*$ are satisfied by the assignment $\pi$ (extended appropriately for~$z'$).

  %
  \item If $\Imc_\Amc,\pi\not\models\exists z'.\beta(z,z')$, then we cannot use a named individual to satisfy the filter $\psi(z)$ in $\Imc_\Kmc$.

  Moreover, since $\Imc_\Amc$ satisfies $\psi(z)$, we also know that $\Imc_\Amc,\pi\not\models\exists z'.\psi^+(z,z')$.
  Since $\psi^+(z,z')=s(z,z')\land N(z')$, this implies that $d\notin(\exists s.N)^{\Imc_\Amc}$.
  Hence, $\exists s.N$ is included in the set~$V$ constructed in Step 3(a) of the canonical model construction for the element $d=\pi(z)$.
  Thus, there exists $M'\sqsubseteq_\Tmc\exists s'.N'$ such that $d\in (M')^{\Imc_\Amc}$, $d\notin(\exists s'.N')^{\Imc_\Amc}$, and $\exists s'.N'\sqsubseteq_\Tmc^s\exists s.N$.
  By Step~3(b), $\Imc_\Kmc$ must contain an element~$d'$ such that, for any $A\in \NC$ and any $r \in \NR$, we have $d'\in A^{\Imc_\Kmc}$ iff $N'\sqsubseteq_\Tmc A$ and $(d,d')\in r^{\Imc_\Kmc}$ iff $s'\sqsubseteq_\Tmc r$.
  Since $N'\sqsubseteq_\Tmc N$ and $s'\sqsubseteq_\Tmc s$, we obtain that $\Imc_\Kmc,\pi\cup\{z'\mapsto d'\}\models\psi^+(z,z')$.

  We show that the assignment $\pi\cup\{z'\mapsto d'\}$ also satisfies $\psi^-(z,z')=\NEG$. Assume to the contrary that there is $\lnot A(z')\in\NEG$ such that $d'\in A^{\Imc_\Kmc}$ (the case of negated role atoms is again analogous).
  Then we have $N'\sqsubseteq_\Tmc A$, which shows that all conditions of~(S3) are satisfied, and hence $M'$ must be included in~$\Mmc'$.
  Since the atoms $\{\neg M'(z) \mid M' \in \mathcal{M}' \}$ are contained in~$\varphi$, we know that they are satisfied by~$\pi$ in~$\Imc_\Kmc$, \ie $d\notin(M')^{\Imc_\Kmc}$ and hence also $d\notin(M')^{\Imc_\Amc}$, which is a contradiction.

  It remains to show that the inner filters~$\Psi^*$ are satisfied by the assignment $\pi\cup\{z'\mapsto d'\}$ in~$\Imc_\Kmc$. Since we are now dealing with an anonymous domain element~$d'$, we can use similar, but simpler, arguments as above to prove this by induction on the structure of the filters. This is possible because the atoms $s(\hat{y}, \hat{x})$, $N(\hat{x})$ implied by $M(\hat{y})$ and the negated atoms induced by~$\Mmc'$ are present in the query even if the filter is integrated into another filter during a subsequent rewriting step.
  \qed
\end{enumerate}
\end{proof*}

We can use this lemma to show correctness of our approach, \ie the answers returned for the \emph{union} of queries given by $\rewo_\Tmc(\phi)$ over~$\Imc_\Amc$ are exactly the answers of the original NCQ~$\phi$ over~$\Imc_\Kmc$.

The proof 
is based on existing proofs for ordinary CQs~\cite{EOS+-AAAI12,bienvenu2015ontology}, extended appropriately to deal with the filters.

\begin{lemma}
  \label{lemma:correctness}
  Let $\Kmc = (\Tmc, \Amc)$ be a consistent $\mathcal{ELH}_\bot$ KB and let $\phi(\mathbf{x})$ be an NCQ.
  Then, for all $\phi'\in\rewo_\Tmc(\phi)$,
  \begin{equation*}
    \mwa(\phi,\Kmc) = \bigcup_{\phi'\in\rewo_\Tmc(\phi)}\ans(\phi',\Imc_\Amc).
  \end{equation*}
\end{lemma}
  \begin{proof*}
   $(\supseteq)$: By \Cref{lemma:filtered-fo-query-data-2}, we have $\ans(\phi',\Imc_\Amc)\subseteq\mwa(\phi',\Kmc)=\ans(\phi',\Imc_\Kmc)$. 

  Furthermore, there exists a sequence $\phi_0 \to_\Tmc \dots \to_\Tmc \phi_n$ $(n > 0)$ with $\phi=\phi_0$ and $\phi' = \phi_n$. Hence it is sufficient to show that
  $\ans(\phi_i,\Imc_\Kmc) \subseteq \ans(\phi_{i-1},\Imc_\Kmc)$ for all $i,1\le i\le n$. 
  Suppose the queries are of the following forms:
  \begin{align}
    \phi_i &= \exists\mathbf{y}_{i}. (\varphi_i(\mathbf{x}_i, \mathbf{y}_i) \land \Psi_i) \\
    \phi_{i-1} &= \exists\mathbf{y}_{i-1}. (\varphi_{i-1}(\mathbf{x}_{i-1}, \mathbf{y}_{i-1}) \land \Psi_{i-1})
  \end{align}

  Let $\pi_i$ be a satisfying assignment for $\varphi_i(\mathbf{x}_i, \mathbf{y}_i) \land \Psi_i$ in $\Imc_\Kmc$.
  Suppose $\phi_{i-1} \to_\Tmc \phi_i$ by
  \begin{enumerate}
    \item selecting variable $\hat{x}$ and introducing $\hat{y}$ in (S1) and
    \item selecting $M \sqsubseteq_\Tmc \exists s. N$ in (S2).
  \end{enumerate}

  Let $\pi_i(\hat{y}) = d$. By Step (S6), $M(\hat{y}) \in \varphi_i$ and since $\pi_i$ satisfies $\varphi_i$, it has to hold that $d \in M^{\Imc_\Kmc}$. This implies that $d \in (\exists s.N)^{\Imc_\Kmc}$.
  Since $\pi_i$ satisfies the new filter~$\psi_i^*(\hat{y})$ that is constructed in~(S7), and by selecting $M \sqsubseteq_\Tmc \exists s. N$ in (S2) 
   the precondition of~$\psi_i^*(\hat{y})$ is satisfied by~$\pi_i$ in~$\Imc_\Kmc$, there has to be an assignment $\pi_i\cup\{\hat{x}\mapsto d'\}$ that satisfies the conclusion of $\psi_i^*(\hat{y})$.

  We define the assignment $\pi_{i-1}$ of the variables of $\varphi_{i-1}$ as follows
  \begin{equation}
    \pi_{i-1}(z) :=
    \begin{cases}
      d' & \text{ if } z = \hat{x} \\
      d  & \text{ if } z \in \PRED \\
      \pi_i(z)  & \text{ otherwise.}
    \end{cases}
  \end{equation}
  Then $\pi_{i-1}$ is a satisfying assignment for $\phi_{i-1}$ in $\Imc_\Kmc$.
  To see this, first consider an atom $\alpha$ in $\varphi_{i-1}$. We show that
  $\pi_{i-1}$ satisfies $\alpha$ in $\Imc_\Kmc$.

  If $\alpha$ contains $\hat{x}$, it can be of the following forms: $A(\hat{x})$, $\neg A(\hat{x})$, $r(y,\hat{x})$ or $\neg r(y,\hat{x})$ with $y \in \PRED$. For all of these cases, we know by Step (S7) that they are either implied by $s(\hat{y},\hat{x})\land N(\hat{x})$ or contained in~$\NEG$, with $y$ replaced by $\hat{y}$.
  By the choice of~$d'$, we know that~$\pi_{i-1}$ satisfies each such atom.

  If $\alpha$ does not contain $\hat{x}$, then $\varphi_i$ contains the atom~$\alpha'$ that is obtained from~$\alpha$ by replacing all of the variables from~$\PRED$ with~$\hat{y}$.
  By construction, we know that $\pi_{i-1}(y) = \pi_{i}(\hat{y})$ for all $y \in \PRED$ and $\pi_{i-1}(z) = \pi_i(z)$ otherwise.
  Since $\alpha'$ is satisfied by~$\pi_i$ in~$\Imc_\Kmc$, $\alpha$ is satisfied by $\pi_{i-1}$ in~$\Imc_\Kmc$.

  What remains to show is that $\pi_{i-1}$ satisfies $\Psi_{i-1}$.
  Consider any $\psi(z) \in \Psi_{i-1}$, and distinguish the following cases:

  \begin{enumerate}
    \item If $z = \hat{x}$, then $\psi(\hat{x}) \in \Psi_{\hat{x}}$. Since  $\Imc_\Kmc, \pi_{i}\cup\{\hat{x}\mapsto d'\} \models \Psi_{\hat{x}}$, we also have $\Imc_\Kmc, \{ \hat{x} \mapsto d' \} \models \psi(\hat{x})$.
    Therefore, since $\pi_{i-1}(\hat{x}) = d'$,
    it holds that $\pi_{i-1}$ satisfies $\psi(\hat{x})$ in~$\Imc_\Kmc$.
    \item If $z \in \PRED$ we know that $\pi_{i-1}(z) = \pi_i(\hat{y}) = d$. Since $\Imc_\Kmc, \pi_{i} \models \psi(\hat{y})$, it also holds that $\Imc_\Kmc, \pi_{i-1} \models \psi(z)$.
    \item Otherwise the filter is present in $\Psi_i$. In this case we know that $\Imc_\Kmc, \pi_i \models \psi(z)$ and $\pi_i(z) = \pi_{i-1}(z)$. Hence, it must also hold that $\Imc_\Kmc, \pi_{i-1} \models \psi(z)$.
  \end{enumerate}

  $(\subseteq)$: Suppose that $\mathbf{a} \in \mwa(\phi,\Kmc) = \ans(\phi,\Imc_\Kmc)$.
  We have to show that there exists a rewriting $\phi' \in \rewo_\Tmc(\phi)$ and a satisfying assignment~$\pi$ for~$\phi'$ in~$\Imc_\Amc$ such that $\mathbf{a} = \pi(\mathbf{x})$.
  To do this, we assign a \emph{degree} (a natural number) to each satisfying assignment (including the existentially quantified variables of the NCQ part) such that a satisfying assignment  with degree~$0$ does not use any anonymous individuals. We then show that for each satisfying assignment with a degree greater than~$0$, we can find a rewriting for which a satisfying assignment yielding the same answer, but with a lower degree, exists.
  In addition, for every such assignment $\pi$ and for all filters $\psi(y)$ in $\phi'$ it should hold that,
  \begin{equation}
    \label{eq:preposition}
    \text{if } \pi(y) \in \NI, \text{ then } \Imc_\Amc \models \psi(\pi(y)),
    \tag{$\dagger$}
  \end{equation}
  \ie all filters (at any stage of the rewriting) are satisfied within the confines of~$\Imc_\Amc$.

  For any element $d \in \Delta^{\Imc_\Kmc}$, we denote by~$|d|$ the minimal number of role connections required to reach~$d$ from an element in~$\NI$, with $|d|=0$ iff $d\in \NI$.
  Additionally, for any assignment $\pi'$ in $\Imc_\Kmc$, let
  \begin{equation}
    \text{deg}(\pi') := \sum\limits_{y \in \text{dom}(\pi')} |\pi'(y)|.
  \end{equation}
  Since $\phi\in \rewo_\Tmc(\phi)$, to prove the claim it suffices to show that whenever there is $\phi_1=\exists\mathbf{y}.\varphi_1(\mathbf{x},\mathbf{y})\land\Psi\in \rewo_\Tmc(\phi)$ such that $\varphi_1$ has a match~$\pi_1$ in $\Imc_\Kmc$ with
  $ \mathbf{a} = \pi_1(\mathbf{x})$, $\text{deg}(\pi_1) > 0$, and \Cref{eq:preposition} holds for $\pi_1$ and the filters in~$\Psi$, then there exist $\phi_2$ and $\pi_2$ with the same properties, but $\text{deg}(\pi_2) < \text{deg}(\pi_1)$.

  Assume $\phi_1 \in \rewo_\Tmc(\phi)$ as above, and let $\pi_1$ be a match of~$\varphi_1$. Since $\text{deg}(\pi_1) > 0$ by assumption, there must exist a variable $\hat{x}$ of $\varphi_1$ such that $\pi_1(\hat{x}) \not\in \NI$.
  Select $\hat{x}$ such that it is a leaf node in the subforest of~$\Imc_\Kmc$ induced by~$\pi_1$. Note that $\hat{x}$ cannot be an answer variable.

  We know that $\pi_1(\hat{x}) = d_{\hat{x}}$ was induced by some axiom $\alpha = M \sqsubseteq_\Tmc \exists s.N$ and element $d_p \in M^{\Imc_\Kmc}$ in Definition~\ref{def:canonical-model}. By the construction of $\Imc_\Kmc$, we know that
  \begin{enumerate}[(i)]
    \item $d_{\hat{x}}$ has just the one predecessor~$d_p$, and
    \item $d_{\hat{x}}\in A^{\Imc_\Kmc}$ iff $N\sqsubseteq_\Tmc A$ and $(d_p,d_{\hat{x}})\in r^{\Imc_\Kmc}$ iff $s\sqsubseteq_\Tmc r$.
  \end{enumerate}
  We obtain the query $\phi_2$ from $\phi_1$ through rewriting, by
  selecting $\hat{x}$ and introducing $\hat{y}$ in (S1), and selecting $\alpha$ in (S2). Let $\PRED$ denote the set of predecessor variables of $\hat{x}$ as defined in (S1).
  To see that this is a valid choice, the conditions in (S2) need to be verified:

  \begin{enumerate}[label=(S2\alph*)]
    \item For any $A(\hat{x}) \in \varphi_1$, we have $d_{\hat{x}}=\pi_1(\hat{x}) \in A^{\Imc_\Kmc}$, and hence $N \sqsubseteq_\Tmc A$ by~(ii).
    Consider any role atom $r(y,\hat{x}) \in \varphi_1$.
    From~(i), the construction of $\Imc_\Kmc$ (no inverse edges), and the fact that $\pi_1$ is a satisfying assignment for $r(y,\hat{x})$ in $\Imc_\Kmc$, the only possibility is that $\pi_1(y) = d_p$.
    Therefore $(d_p, d_{\hat{x}}) = (\pi_1(y),\pi_1(\hat{x})) \in r^{\Imc_\Kmc}$.
    By~(ii), this implies that $s \sqsubseteq_\Tmc r$.
    \item Consider any $\neg A(\hat{x}) \in \varphi_1$, for which we must have $d_{\hat{x}}\notin A^{\Imc_\Kmc}$.
    From (ii) we know that $N \not\sqsubseteq_\Tmc A$.
    Consider any $\neg r(y,\hat{x}) \in \varphi_1$. Since this is guarded by a positive role atom as above, again the only possibility is that $\pi_1(y) = d_p$.
    Hence $(d_p, d_{\hat{x}}) \notin r^{\Imc_\Kmc}$.
    By~(ii), this implies that $s \not\sqsubseteq_\Tmc r$.
  \end{enumerate}

  Therefore, we obtain a satisfying assignment $\pi_2$ for $\phi_2$ in $\Imc_\Kmc$ such that $\mathbf{a} \in \pi_2(\mathbf{x})$ (and $\text{deg}(\pi_2) < \text{deg}(\pi_1)$) by setting for all $z \in \Var(\varphi_2)$:
  $$\pi_2(z) := \begin{cases}
    \pi_1(z) & \text{ if } z \in \Var(\varphi_1) \\
    d_p      & \text{ if } z = \hat{y}.
  \end{cases}$$
  To see that $\pi_2$ satisfies~$\phi_2$, we argue why it satisfies the new atoms and filter from~(S6) and~(S7); the old atoms (possibly with renamed variables) remain satisfied.

  The new atom $M(\hat{y})$ is satisfied since $\pi_2(\hat{y})=d_p\in M^{\Imc_\Kmc}$.
  Consider now an atom $\lnot M'(\hat{y})$ with $M'\in\Mmc'$ as specified in~(S6);
  we have to show that $d_p\notin (M')^{\Imc_\Kmc}$.
  Assume to the contrary that $d_p\in (M')^{\Imc_\Kmc}$.
  By~(S3), we know that $M' \sqsubseteq_\Tmc \exists s'.N' \sqsubseteq_\Tmc^s \exists s.N$.
  Moreover, $\exists s'.N'$ must be included in the set~$V$ in Step~3(a) of Definition~\ref{def:canonical-model}, because otherwise we would already have
  $d_p\in(\exists s'.N')^{\Imc_\Amc}$, \ie there would be a named individual $b$ such that $(d_p,b)\in (s')^{\Imc_\Amc}$ and $b\in (N')^{\Imc_\Amc}$.
  Since $s'\sqsubseteq_\Tmc s$ and $N'\sqsubseteq_\Tmc N$, this would imply $(d_p,b)\in s^{\Imc_\Amc}$ and $b\in N^{\Imc_\Amc}$, \ie $d_p\in(\exists s.N)^{\Imc_\Amc}$, which shows that the anonymous object~$d_{\hat{x}}$ would not have been created.
  Since $\exists s'.N'$ is included in~$V$ and we assumed that $\exists s.N$ is minimal \wrt $\sqsubseteq_\Tmc^s$, we must have $s\equiv_\Tmc s'$ and $N\equiv_\Tmc N'$.
  But then (S3b) directly contradicts (S2b).

  We now consider the filters in~$\phi_2$.
  Suppose that \Cref{eq:preposition} holds for~$\pi_1$ and all filters in~$\phi_1$.
  For the ones that are only copied from~$\phi_1$ (modulo renaming some variables to~$\hat{y}$), the property is clearly preserved.
  %
  For the new filter $\psi^*(\hat{y})$, assume that $\pi_2(\hat{y}) \in \NI$, and hence we need to show that $\Imc_\Amc\models\pi_2(\psi^*(\hat{y}))$.
  Assume that there exists an element $d' \in \NI$ such that $(d_p, d') \in s^{\Imc_\Amc}$ and $d' \in N^{\Imc_\Amc}$.
  But then in Step 3(a) in \Cref{def:canonical-model}, $\exists s.N$ could not have been added to~$V$ since $d_p \in (\exists s.N)^{\Imc_\Kmc}$ already holds. 
  Hence, the element $d_{\hat{x}}$ would have never been introduced, which is a contradiction. Therefore, in $\Imc_\Amc$ the precondition of $\psi^*(\hat{y})$ is never met, which makes the filter trivially satisfied.

  Finally, to show that $\text{deg}(\pi_2) < \text{deg}(\pi_1)$, we make a case distinction on whether the set $\PRED$ is empty or not.
  If $\PRED=\emptyset$, then we essentially replace the variable~$\hat{x}$ in~$\varphi_1$ with a new variable~$\hat{y}$ in~$\varphi_2$ with $|\pi_2(\hat{y})|=|d_p|<|d_{\hat{x}}|=|\pi_1(\hat{x})|$.
  Since the remaining variables are not affected by the rewriting step, this shows that $\text{deg}(\pi_2) < \text{deg}(\pi_1)$.
  If $\PRED\neq\emptyset$, then we have $|\pi_2(\hat{y})|=|d_p|=|\pi_1(y)|$ for all $y\in\PRED$.
  Since the variables in $\Var(\varphi_1)\setminus\{\hat{y}\}=\Var(\varphi_1)\setminus(\PRED\cup\{\hat{x}\})$ are not affected and $|\pi_1(\hat{x})|>0$, we conclude that
  \begin{align*}
    \text{deg}(\pi_2)
    &= |\pi_2(\hat{y})|+\sum_{z\in\Var(\varphi_2)\setminus\{\hat{y}\}}|\pi_2(z)| \\
    &< |\pi_1(\hat{x})|+\sum_{y\in\PRED}|\pi_1(y)|+\sum_{z\in\Var(\varphi_1)\setminus(\PRED\cup\{\hat{x}\})}|\pi_1(z)| \\
    &= \text{deg}(\pi_1). \qed
  \end{align*}
  \end{proof*}

Under data complexity assumptions, $\phi$ and~\Tmc, and hence $\rewo_\Tmc(\phi)$, are fixed, and $\Imc_\Amc$ is of polynomial size in the size of~\Amc.

For non-rooted queries the rewriting can get infinite, because we can always find a leaf variable that is not an answer variable.
For example, suppose that the TBox $\Tmc$ consists of the two GCIs $A \sqsubseteq \exists r.B$ and $B \sqsubseteq \exists r.A$.
Let $\phi() = \exists x. A(x) \land \neg B(x)$ be a Boolean query.
The rewriting algorithm would produce an infinite rewriting of the form:
\begin{align*}
  \phi_0 &= \exists x. A(x) \land \neg B(x) \\
  \phi_1 &= \exists x. B(x) \land \underbrace{\big(\exists y. r(x,y) \land A(y) \to \exists y. r(x,y) \land A(y) \land \neg B(y) \big)}_{\psi_1(x)} \\
  \phi_2 &= \exists x. A(x) \land \underbrace{\big(\exists y. r(x,y) \land B(y) \to \exists y. r(x,y) \land B(y) \land \psi_1(y) \big)}_{\psi_2(x)} \\
  \phi_3 &= \exists x. B(x) \land \underbrace{\big(\exists y. r(x,y) \land A(y) \to \exists y. r(x,y) \land A(y) \land \psi_2(y) \big)}_{\psi_3(x)} \\
  \dots
\end{align*}

In the following, we show that such infinite rewritings can be avoided, since after a certain number of rewriting steps queries will not generate any new answers.

For a query $\phi$ with filters $\Psi$, where each $\psi(y) \in \Psi$ is of the form
$$\big( \exists y'. \psi^+(y,y') \big) \to \big(\exists y'. \psi^+(y,y') \land \psi^-(y,y') \land \Psi' \big),$$
we define the \emph{nested filter depth} as follows:
\begin{align*}
  |\phi| &:= \max\limits_{\psi \in \Psi} |\psi| & |\psi| &:= 1  + \max\limits_{\psi' \in \Psi'} |\psi'|,
\end{align*}
where the second expression is applied recursively to subfilters.

\begin{lemma}
  \label{lemma:non-rooted-bound}
  Let $\Kmc = (\Tmc, \Amc)$ be a consistent $\ELHb$ KB and let $\phi$ be an NCQ.
  Then
  \begin{equation*}
    \bigcup_{\substack{\phi'\in\rewo_\Tmc(\phi) \\ |\phi'| \leq v + \NC_\Tmc^2\cdot \NR_\Tmc}} \ans(\phi',\Imc_\Amc) = \bigcup_{\phi'\in\rewo_\Tmc(\phi)} \ans(\phi',\Imc_\Amc).
  \end{equation*}
  where $v$ denotes the number of variables in $\phi$, and $\NC_\Tmc$ and $\NR_\Tmc$ denote the number of concept and role names in $\Tmc$, respectively.
\end{lemma}

\begin{proof}

In the following, we assume a connected query $\phi$ with $v$ variables.
This is without loss of generality since the non-connected parts in the query can be dealt with separately.

Connectedness is preserved by the rewriting algorithm and hence there are two possible scenarios.
In the first one, $\phi$ has only a finite number of rewritings. That happens if after at most $v$ rewriting steps the rewriting is of the form $\exists \mathbf{z}.\varphi(\mathbf{z}) \land \Psi(\mathbf{z})$, where $\varphi(\mathbf{z})$ cannot be rewritten further.

Otherwise, $\phi$ has infinitely many rewritings. In this case, after $v$ rewriting steps, the rewriting and all further rewritings are of the form
\begin{equation}
    \label{eq:linear-form}
    \tag{$\ddagger$}
    \exists y. \big(A(y) \land \NEG \land \, \psi(y)\big),
\end{equation}
where $A \in \NC$ (recall that we assume a TBox in normal form), $\NEG$ is a conjunction of atoms of the form $\neg \hat{A}(y)$ for $\hat{A}\in \NC$, and $\psi$ is a filter that is linearly extended in every further rewriting step.

Assume a query $\phi_0$ of the form \eqref{eq:linear-form} that has been further rewritten to $\phi_n$ in a sequence of
$\phi_0 \to_\Tmc \dots \to_\Tmc \phi_i \to_\Tmc \dots \to_\Tmc \phi_n$, where $n \in \Nbb$, every $\phi_j$ is of the form $\exists y.A_j(y) \land  \NEG_j \land\, \psi_j(y)$ for $0 \leq j \leq n$, and there exists $1 < i < n$ with $A_n = A_i$, $\NEG_n = \NEG_i$ and $\psi_n^+ = \psi_i^+$, \ie the filters of $\phi_n$ and $\phi_i$ have the same body (up to renaming of variables).

We show that if there is a satisfying assignment $\pi_n$ for $\phi_n$ over $\Imc_\Amc$, then there is also a satisfying assignment for a query $\phi'$ with $|\phi'| < |\phi_n|$.

Suppose the nested filters in $\phi_n$ are matched up to a nested filter depth of $0 \leq d \leq n$.

Case 1: If $d = 0$, the body of $\psi_n$ is not satisfied, then $\pi_n$ is also a match for $A_i\land \NEG_i$, since the $A_i = A_n$, $\NEG_i = \NEG_n$ and the body of $\psi_i$ is the same as the body of $\psi_n$ by assumption. Clearly $|\phi_i| < |\phi_n|$.

Case 2: If $d > 0$, then we construct a match for query $\phi_{n-1}$.
Suppose $\phi_{n-1}$ has been rewritten to $\phi_n$ by using a GCI $A \sqsubseteq_\Tmc \exists r.B$. Then we know that
\begin{align*}
  \phi_{n-1} &:= \exists y. B(y) \land \NEG_{n-1} \land\, \psi(y) \\
  \phi_n &:= \exists x. A(x) \land \NEG_n \land\, \big( \exists y. r(x,y) \land B(y) \to \exists y. r(x,y) \land B(y) \land \NEG_{n-1} \land\, \psi(y) \big)
\end{align*}
Because $d > 0$, there has to be a satisfying assignment $\{x \mapsto a, y \mapsto b \}$ with $a,b \in \Delta^{\Imc_\Amc}$
for $r(x,y) \land B(y)$, which also satisfies $\NEG_{n-1}\land\,\psi(y)$. Then the assignment $\pi_{n-1} := \{y \mapsto b\}$ satisfies~$\phi_{n-1}$, and $\psi(y)$ will then be matched up to a depth of $d-1$.

This result can be used to bound the nested filter depth of queries during rewriting:
In each rewriting step the query is rewritten \wrt a GCI of the form $A \sqsubseteq_\Tmc \exists r. B$ with $A,B \in \NC_\Tmc$ and $r \in \NR_\Tmc$.

There can be at most $\NC_\Tmc^2 \cdot \NR_\Tmc$ different GCIs of this form.
Suppose a rewriting of $\phi$ to $\phi'$ with $|\phi'| > v + \NC_\Tmc^2 \cdot \NR_\Tmc$, where $v$ denotes the number of variables in~$\phi$.
Then at least two rewritings between $\phi$ and~$\phi'$ must start with the same expression $\exists x.A(x)\land \NEG\land(\exists y.r(x,y)\land B(y) \to \dots )$.

By the arguments above, there are rewritings of $\phi$ of nested filter depth at most $v + \NC_\Tmc^2 \cdot \NR_\Tmc$ that yield the same answers as $\phi'$.

\end{proof}

Hence, queries need to be rewritten only until this bound on the nested filter depth, which does not depend on~\Amc.
We obtain the claimed complexity result.

\begin{theorem}
  \label{theorem:complexity}
  Answering NCQs under minimal-world semantics over consistent \ELHb KBs is \NP-complete, and \PTime-complete in data complexity.
\end{theorem}
\begin{proof}
  The lower bounds are inherited from certain answer semantics of CQs over \EL KBs \cite{Rosa-DL07} since for CQs the two semantics coincide.
  For data complexity, it suffices to observe that the size and number of the rewritten queries does not depend on the ABox, and the size of~$\Imc_\Amc$ is polynomial in the size of~\Amc, hence evaluating the rewriting over~$\Imc_\Amc$ can be done in polynomial time.

  For combined complexity, first we can guess an element~$\phi'$ of~$\rewo_\Tmc(\phi)$ since by Lemma~\ref{lemma:non-rooted-bound} the number of necessary rewriting steps is bounded polynomially and each rewriting is of polynomial size.
  Now it remains to evaluate~$\phi'$ over~$\Imc_\Amc$, which is possible in \PSpace since $\phi'$ is a first-order query.
  To obtain a bound of~\NP, we need to do a more careful evaluation.
  First, we start by guessing a match of the initial CQ part of~$\phi'$ in~$\Imc_\Amc$.
  For each of the polynomially many (nested) filters
  \[ \big(\exists z'. \psi^+(z,z')\big) \to \big(\exists z'. \psi^+(z,z') \land \psi^-(z,z') \land \Psi\big) \]
  in~$\phi'$, we then enumerate all (polynomially many) possibilities to extend the initial match by mapping~$z'$ into~$\Imc_\Amc$.
  If none of these satisfies~$\psi^+(z,z')$, this (deterministic) check ends the evaluation successfully.
  Otherwise, we can continue with recursively guessing an extended match for $\exists z'. \psi^+(z,z') \land \psi^-(z,z') \land \Psi$, for each nested filter extending the match by one new variable.
  Overall, this takes only polynomial time, which implies the claim.
\end{proof}

As discussed in \Cref{example:cancer-patient} it can be useful to allow \emph{complex assertions} in the ABox:
The assertions of patient $\SkinOfBreastPatient$ can then be stated without introducing a new constant for the disease by stating
$$\exists \text{diagnosedWith}.\text{SkinOfBreastCancer}(\SkinOfBreastPatient).$$
This leads to the introduction of additional acyclic definitions~$\Tmc'$, which are not fixed. The complexity nevertheless remains the same: Since \Tmc does not use the new concept names in~$\Tmc'$, we can apply the rewriting only w.r.t.~\Tmc, and extend~$\Imc_\Amc$ by a polynomial number of new elements that result from applying \Cref{def:canonical-model} only w.r.t.~$\Tmc'$. 

What is more important than the complexity result is that this approach can be used to evaluate NCQs using standard database methods, \eg using views to define the finite interpretation~$\Imc_\Amc$ based on the input data given in~\Amc, and SQL queries to evaluate the elements of $\rewo_\Tmc(\phi)$ over these views \cite{KLT+-IJCAI11}.

\section{\texorpdfstring{Temporal Minimal World Semantics in \dELHbm}{Temporal Minimal World Semantics in TELH_bot\textasciicircum c,lhs,-}}
In the previous part, we introduced the minimal-world semantics and showed how NCQs can be answered efficiently over \ELHb-KBs.
In the remaining part of the paper we extend the minimal-world semantics to answering temporal queries over \dELHbm.

\begin{example}
  \label{ex:patients}
  Continuing the previous example of different cancer patients, we can now also model temporal aspects.
  \emph{Chemotherapy} is a treatment that is usually given in a number of cycles to fight the cancer of a given patient.
  We formalize this knowledge as follows, where the basic time unit is one day:
  \begin{align}
    \ChemoP &\sqsubseteq \CancerP \label{eq:ax0}\\
    \diamondc_{365} \CancerP &\sqsubseteq \CancerP \label{eq:ax1}\\
    \diamondc_{120} \ChemoP &\sqsubseteq \ChemoP \label{eq:ax3}
  \end{align}
  We make the assumption that being a cancer patient is 365-days convex, hence if a patient is reported to have cancer at two time points at most 1 year apart, we assume that they were also suffering from cancer in between the two reports.
  Similarly, we assume that a patient reported to be receiving chemotherapy at most four months apart also received chemotherapy in between.

  Suppose the ABox~\Amc contains assertions $\ChemoP(p_1, i)$, $i\in\{0,167,258\}$, for a patient~$p_1$.
  The set of entailed assertions, denoted by~$\Amc^*$, is illustrated below, where for simplicity we omit the individual name~$p_1$.
   Here, \CancerP and \ChemoP are abbreviated by C and T, respectively. Representative time points $-1$, $60$, $200$, and $259$ have been introduced additionally, and the intervals they represent are shown above them.
      %
  \begin{center}

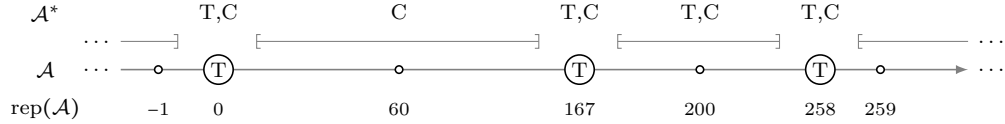
\newcommand\PMIN{-1}
\newcommand\PMAX{11}
\begin{tikzpicture}[>=latex, yscale=0.75, xscale=1, semithick,
point/.style={circle,fill=white,draw=black,inner sep=1pt}]\footnotesize

\draw[draw=gray] (\PMIN-0.5,0) -- ++(\PMAX+0.3,0);
\node (ldots) at (\PMIN-0.8,0) {\dots};
\node (rdots) at (\PMAX-0.9,0) {\dots};
%
%

\foreach \name/\position/\timepoint/\asserted/\completed in 
  {v-1/-1/-1/{}/{},
   p0/0/0/{T}/{T,C},
   v1/3/60/{}/{C},
   p2/6/167/{T}/{T,C},
   v3/8/200/{}/{T,C},
   p4/10/258/{T}/{T,C},
   v5/11/259/{}/{}} {
  \node (\name) at (-0.2+\position*0.8,0) [point]{\asserted};
  \node at (-0.2+\position*0.8,-0.7) {{\scriptsize $\timepoint$}};
  \node[rotate=0, anchor=north] at ($(\name)+(0,1.3)$) {{\completed}};
  }
  
\node at (\PMIN-1.5,1.0) {$\Amc^*$};
\node at (\PMIN-1.5,-0.7) {$\zaplus$};
\node at (\PMIN-1.5,0) {$\Amc$};

\node at ($(ldots)+(0,0.5)$) {\dots};
\draw[{-]},thin,gray]	($(ldots)+(0.3,0.5)$) -- ($(v-1)+(0.3,0.5)$); 
\draw[{[-]},thin,gray]	($(p0)+(0.5,0.5)$) -- ($(p2)+(-0.5,0.5)$);
\draw[{[-]},thin,gray]	($(p2)+(0.5,0.5)$) -- ($(p4)+(-0.5,0.5)$);
\node at ($(rdots)+(0,0.5)$) {\dots};
\draw[{[-},thin,gray]	($(v5)+(-0.3,0.5)$) -- ($(rdots)+(-0.3,0.5)$);

%
%

\end{tikzpicture}
  \end{center}
  Since there are infinitely many time points, in general $\Amc^*$ is infinite. However, in~\cite{BoFoKo-RRML} this infinite set is represented via a finite set of \emph{representative time points} $\zaplus$, each of which represents a time interval in which the assertions do not change. 
\end{example}

\subsection{Metric Temporal Conjunctive Queries with Negation}

We now consider the reasoning problem of temporal query answering, which generalizes the results from Section~\ref{sec:atemp-rewriting}. We develop a new temporal query language with negation and extend the closed-world semantics for negation from the previous sections to \dELHbm.

We combine atemporal NCQs with MTL operators~\cite{AlHe-JACM94,GuJO-ECAI16,BBK+-FroCoS17} to obtain Metric Temporal NCQs (MTNCQs) and finally show that such queries can be answered efficiently over \dELHbm KBs when using the minimal-world semantics.

\begin{definition}
\emph{Metric temporal conjunctive queries with negation (MTNCQs)} are built by the grammar rule
\begin{equation}
  \phi::= \psi \mid \top \mid \bot \mid \neg \phi \mid \phi \land \phi \mid \phi \lor \phi 
\mid \phi \luntil_{I} \phi \mid \phi \lsince_{I} \phi,
\end{equation}
where $\psi$ is an NCQ, and $I$ is an interval over $\Nbb$.

An MTNCQ $\phi$ is \emph{rooted}/\emph{Boolean }if all NCQs in it are rooted/Boolean.

\end{definition}

We employ the standard semantics shown in~\Cref{fig:tcq-semantics}.
One can define the \emph{next} operator as $\lnext \phi := \top\luntil_{[1,1]}\phi$, and similarly $\lprev \phi := \top\lsince_{[1,1]}\phi$.
We can also express
$$\leventually_I\phi := (\top\lsince_{-(I\cap(-\infty,0])}\phi)\lor(\top\luntil_{I\cap[0,\infty)}\phi)$$ and $\lalways_I \phi := \lnot\leventually_I\lnot\phi$,
and hence, by~\eqref{eq:convex-diamond}, the $\diamondc[n]$-operators from Section~\ref{sec:delhb}.
An \emph{MTCQ} (or \emph{positive MTNCQ}) is an MTNCQ without negation, where we assume that the operator~$\lalways_I$ is nevertheless included as part of the syntax of MTCQs.

\begin{figure}[t]
 \centering
 \setlength{\tabcolsep}{10pt}
\begin{tabular}{lp{8cm}}
  \toprule
  $\phi$ & $\Imf, i \models \phi$ iff \\
  \midrule
  CQ~$\psi$ & $\Imc_i \models \psi$ \\
  $\top$ & true \\
  $\bot$ & false \\
  $\neg \phi$ & $\Imf, i \not\models \phi$ \\
  $\phi \land \psi$ & $\Imf, i \models \phi \text{ and } \Imf, i \models \psi$ \\
  $\phi \lor \psi$ & $\Imf, i \models \phi \text{ or } \Imf, i \models \psi$ \\
  $\phi\luntil_{I} \psi$ & $\exists k \in I$ such that $\Imf, i+k \models \psi$ and $\forall j: 0 \leq j < k: \Imf, i+j \models \phi$ \\
  $\phi\lsince_{I} \psi$ & $\exists k\in I$ such that $\Imf, i-k \models \psi$ and $\forall j: 0 \leq j < k: \Imf, i-j \models \phi$ \\
  \bottomrule
\end{tabular}

\caption{Semantics of (Boolean) MTNCQs for $\Imf=(\Delta^\Imf,(\Imc_i)_{i\in\Zbb})$ and $i\in\Zbb$.}
\label{fig:tcq-semantics}
\end{figure}

Let $\Kmc = (\Tmc, \Amc)$ be a \dELHbm KB, $\phi(\xbf)$ an MTNCQ, \abf a tuple of individual names from~\Amc, $i\in\tem$, and \Imf a temporal interpretation.

The pair $(\abf, i)$ is an \emph{answer} to $\phi (\xbf)$ w.r.t.~\Imf if $\Imf,i \models \phi(\abf)$. The set of all answers for $\phi$ w.r.t.~\Imf is denoted $\ans(\phi,\Imf)$.
The tuple $(\abf,i)$ is a \emph{certain answer} to~$\phi$ w.r.t.~$\Kmc$ if it is an answer in every model of~$\Kmc$; all these tuples are collected in the set $\cert(\phi,\Kmc)$.

\begin{example}
  \label{ex:chemotherapy}
  Consider the criterion ``Patients that received chemotherapy for more than 3 months, but less than 6 months.''
  This can be expressed as an MTNCQ as follows.
  \begin{align*}
      \phi(x) :={} &\lalways_{[-90, 0]} \ChemoP(x)
             \land \neg \lalways_{[-180, 0]} \ChemoP(x)
  \end{align*}
  The negated conjunct expresses that the data does not indicate an ongoing chemotherapy for the past 6 months (minimal-world semantics), rather than that such a treatment is categorically ruled out by some TBox axioms (certain-answer semantics).
\end{example}

\subsection{Minimal-World Semantics for MTNCQs}
\label{sec:temporal-ncqs}

To extend the approach from~Section~\ref{sec:atemp-rewriting} we need to find a \emph{minimal canonical model} of a \dELHbm KB.
%

%

Note that $\dELHbm$ does not allow temporal roles, because temporal roles interfere with the \emph{minimality}:
they may entail the existence of a generic role successor long after it has been superseded by several more specific successors, and hence the generic successor becomes redundant.

In the definition of the model, we make use of entailment in \dELHbm, which can be checked in polynomial time~\cite{BoFoKo-RRML}. Thus, we can exclude w.l.o.g.\ equivalent concept and role names.
Also, for simplicity, in the following we assume w.l.o.g.\ that all TBox inclusions are in normal form~\eqref{eq:normal1}.

In particular, disallowing CIs of the form $\da A\sqsubseteq\exists r.B$ allows us to draw a stronger connection to the original construction in~Section~\ref{sec:atemp-rewriting}; see in particular Step~3(a) in Definition~\ref{def:min-tem-canonical} below.

\begin{definition}\label{def:min-tem-canonical}
  The \emph{minimal temporal canonical model} $\Imf_\Kmc=(\Delta^{\Imf_\Kmc},(\Imc_i)_{i\in \Zbb})$ of a KB 
 $\Kmc = (\Tmc,\Amc)$ is obtained by the following steps.
  \begin{enumerate}
    \item\label{c-domain} Set $\Delta^{\Imf_\Kmc} := \NI$ and $a^{\Imc_{i}} := a$ for all $a \in \NI$ and $i\in \Zbb$.
    \item\label{c-abox} For each time point $i \in \Zbb$, define $A^{\Imc_i} := \{ a \mid \Kmc \models A(a,i) \}$ for all $ A \in \NC$ and $r^{\Imc_i} := \{ (a,b) \mid \Kmc \models r(a,b,i) \}$ for all $ r \in \NR$.
    \item\label{c-successors} Repeat the following steps:
    \begin{enumerate}
      \item\label{c-select-d} Select an element $d \in \Delta^{\Imf_\Kmc}$ that has not been selected before and, for each $i\in \Zbb$,
      let $V_i := \{ \exists r.B \mid d \in A^{\Imc_i},\ d \not\in (\exists r. B)^{\Imc_i},$ $\Kmc\models A \sqsubseteq \exists r.B,\ A,B\in \NC \}$.
       \item\label{c-add-e} For each $\exists r.B$ that is minimal in some~$V_i$, add a fresh element $e_{rB}$ to $\Delta^{\Imf_\Kmc}$. For all $i\in \Zbb$ and $\Kmc\models B\sqsubseteq A$, add $e_{rB}$ to $A^{\Imc_i}$.
      \item\label{c-add-e-2} For all $i\in \Zbb$, minimal $\exists r.B$ in~$V_i$, and $\Kmc\models r\sqsubseteq s$, add $(d,e_{rB})$ to $s^{\Imc_i}$.
    \end{enumerate}
  \end{enumerate}
\end{definition}
We denote by $\aboxi$ the result of executing only Steps~1 and~2 of this definition, \ie restricting $\Imf_\Kmc$ to the named individuals.
Since there are only finitely many elements of~\NI, \NC, and \NR that are relevant for this definition (\ie those that occur in~\Kmc), for simplicity we often treat $\aboxi$ as if it had a finite object (but still infinite time) domain.

In~$\Imf_\Kmc$, there may exist anonymous objects that are not connected to any named individuals in~$\Imc_i$ and are not relevant for the satisfaction of the KB.
For this reason, in the following we consider only rooted MTNCQs, which can be evaluated only over the parts of~$\Imf_\Kmc$ that are connected to the named individuals.
We again show that~$\Imf_\Kmc$ is actually a model of~\Kmc and is canonical in the usual sense that it can be used to answer \emph{positive} queries over~\Kmc under certain answer semantics.

\begin{restatable}{lemma}{modelconstructioncorrectness}
\label{lem:model-construction-correctness}
  Let \Kmc be a consistent \dELHbm KB. Then ${\Imf_\Kmc}$ is a model of $\Kmc$ and, for every rooted MTCQ~$\phi$, we have $\cert(\phi,\Kmc)=\ans(\phi,\Imf_\Kmc)$.
\end{restatable}

  \begin{proof*}
    The first claim is easy to prove, and the inclusion $\cert(\phi,\Kmc)\subseteq\ans(\phi,\Imf_\Kmc)$ follows from the fact that $\Imf_\Kmc$ is a model of~\Kmc. For the other inclusion, consider any model $\Jmf=(\Delta^\Jmf,(\Jmc_i)_{i\ge 0})$ of~\Kmc.
    We prove that $(\abf,i)\in\ans(\phi,\Imf_\Kmc)$ implies $(\abf,i)\in\ans(\phi,\Jmf)$ by induction on the structure of~$\phi$.
    \begin{itemize}
      \item If $\phi$ is a rooted CQ, then $\Imc_i\models\phi(\abf)$. Moreover, since $\phi$ is rooted, only the rooted part of $\Imc_i$, consisting of all elements connected to named individuals, is relevant for satisfying~$\phi(\abf)$. It is easy to show that this part can be homomorphically mapped into~$\Jmc_i$, hence $\Jmf,i\models\phi(\abf)$.
      \item If $\phi=\phi_1\lor\phi_2$, then $(\abf,i)\in\ans(\phi_1,\Imf_\Kmc)$ or $(\abf,i)\in\ans(\phi_2,\Imf_\Kmc)$, hence by induction $(\abf,i)\in\ans(\phi_1,\Jmf)$ or $(\abf,i)\in\ans(\phi_2,\Jmf)$, either of which implies that $(\abf,i)\in\ans(\phi,\Jmf)$.
      \item The cases of $\lsince_I$, $\leventually_I$, and $\lalways_I$ are similar, and therefore the claim also extends to $\diamondc[n]$, $\lnext$, and $\lnext^-$.
      \qed
    \end{itemize}
  \end{proof*}

Thus, the following \emph{minimal-world} semantics is compatible with certain answer semantics for positive (rooted) queries, while keeping a tractable data complexity.

\begin{definition}
  The set of \emph{minimal-world answers} to an MTNCQ~$\phi$ over a consistent \dELHbm KB~\Kmc is $\mwa(\phi,\Kmc):=\ans(\phi,\Imf_\Kmc)$.
\end{definition}

\section{A Combined Rewriting for MTNCQs}

Following the approach used for the atemporal case, we now show that rooted MTNCQ answering under minimal-world semantics is combined first-order rewritable. For rewriting atomic queries we use the results from~\cite{BoFoKo-RRML}. Thus, we proceed as follows. First we rewrite $\phi$ into a \emph{metric first-order temporal logic (\MFOTL)} formula~$\rew{\phi}$, which combines first-order formulas via metric temporal operators; for details, see~\cite{BKMZ-JACM15}. The formula $\rew{\phi}$ is obtained by replacing each (rooted) NCQ~$\psi$ with the first-order rewriting~$\rew{\psi}:=\bigvee_{\psi'\in\rewo_\Tmc(\psi)}\psi'$
from Section~\ref{sec:atemp-rewriting}, and hence $\rew{\phi}$ can be evaluated already over $\aboxi$.
Second, we then further rewrite $\rew{\phi}$ into a three-sorted

first-order formula (with explicit variables for (a)~objects, (b)~time points, and (c)~bits of the binary representation of time values), which is then evaluated over a restriction~$\aboxf$ of~$\aboxi$ that contains only finitely many time points. 

These two steps produce a valid rewriting for the query $\phi$.

\begin{restatable}{lemma}{rewritinginifinite}
  \label{lem:rewritinginfinite}
  Let $\Kmc = (\Tmc, \Amc)$ be a consistent \dELHbm KB and $\phi$ be a rooted MTNCQ.
  Then
  $\mwa(\phi, \Kmc) = \ans(\rew{\phi}, \aboxi)$.
\end{restatable}
  \begin{proof*}
    We prove the claim by induction over the structure of~$\phi$.

    Suppose that $\phi$ is a rooted NCQ.
    Since $\phi$ does not contain temporal operators, we can restrict our attention to a single atemporal interpretation~$\Imc_i$ in $\Imf_\Kmc=(\Delta^{\Imf_\Kmc},(\Imc_i)_{i\in\Zbb})$.
    Since $\phi$ is rooted, only the \enquote{rooted} part~$\Imc_i^r$ of~$\Imc_i$, consisting only of those elements connected to named individuals via a sequence of role connections, is relevant for evaluating~$\phi$ (and similarly for $\rew{\phi}$).
    In the construction of $\Imf_\Kmc$ (\Cref{def:min-tem-canonical}), we can observe that $\Imc_i^r$ is uniquely determined by the definition of $A^{\Imc_i}$ and $r^{\Imc_i}$ in Step~\ref{c-abox}. Moreover, $\Imc_i^r$ is isomorphic to the (atemporal) minimal canonical model of~$(\Tmc',\Amc_i)$ as in Definition~\ref{def:canonical-model}, where
    \begin{itemize}
      \item $\Tmc':=\{C\sqsubseteq D \mid \da C\sqsubseteq D\in\Tmc,\da\in\Dpm\}$ and
      \item $\Amc_i:=\{A(a)\mid\Kmc\models A(a,i)\}\cup\{r(a,b)\mid\Kmc\models r(a,b,i)\}$.
    \end{itemize}
    In particular, one can observe that the temporal operators in~\Tmc are irrelevant for the behavior of the anonymous elements in~$\Imc_i$ (note that $\da C\sqsubseteq D$ entails $C\sqsubseteq D$) and we can restrict the attention to those assertions entailed for time point~$i$.
    Hence, by Lemma~\ref{lemma:correctness} we can conclude that
    $(\abf,i)\in\mwa(\phi,\Kmc)=\ans(\phi,\Imf_\Kmc)$ iff $\abf\in\ans(\phi,\Imc_i)=\ans(\phi,\Imc_i^r)=\ans(\rew{\phi},\Imc_{i,\Amc})$ iff $(\abf,i)\in\ans(\rew{\phi},\Imf_\Amc)$, where $\Imc_{i,\Amc}$ is the restriction of~$\Imc_i$ to the named individuals.

    For the remaining cases, it suffices to observe that $\phi$ and $\rew{\phi}$ are built on the same structure of temporal operators, which have the same semantics for both MTNCQs and MFOTL formulas.
    \qed
  \end{proof*}

For the second rewriting step, we restrict ourselves to finitely many time points.
More precisely, we consider the finite structure~$\aboxf$, which is obtained from~$\aboxi$ by restricting the set of time points to~$\zaplus$.
By Lemma~4 in~\cite{BoFoKo-RRML}, the information contained in $\aboxf$ is already sufficient to answer rooted atomic queries.
The idea is that each time point $i$ is assigned a representative time point $|i|$ from~$\zaplus$, at which the same assertions hold. Therefore, the answers to a query at time point $i$ are the same as the answers at $|i|$.

We extend this structure a little, by considering the \emph{two} representatives~$i,j$ for each maximal interval $[i,j]$ in $\Zbb\setminus\tem$.
In this way, we ensure that the \enquote{border} elements are always representatives for their respective intervals.
The size of the resulting structure~$\aboxf$ is polynomial in the size of~\Kmc.

\begin{example}\label{ex:period1}

Continuing \Cref{ex:chemotherapy}, below one can see the finite structure~$\aboxf$ over the representative time points $\{-1,0,1,166,167,168,257,258,259\}$, where for simplicity we omit the individual name.
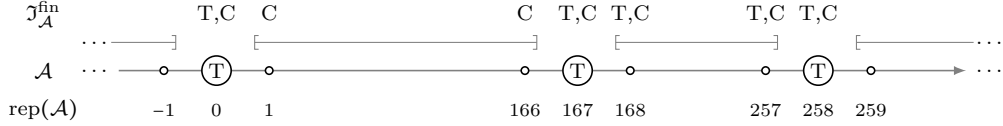
\begin{center}

\newcommand\PMIN{-1}
\newcommand\PMAX{11}
\newcommand\DISTX{0.1}
\begin{tikzpicture}[>=latex, yscale=0.75, xscale=1, semithick,
point/.style={anchor=center,circle,fill=white,draw=black,inner sep=1pt}]\footnotesize

\draw[draw=gray] (\PMIN-0.5,0) -- ++(\PMAX+0.3,0);
\node (ldots) at (\PMIN-0.8,0) {\dots};
\node (rdots) at (\PMAX-0.9,0) {\dots};
%
%

\foreach \name/\position/\offset/\timepoint/\asserted/\completed in 
  {v-1/-1/+\DISTX/-1/{}/{},
   p0/0/0/0/{T}/{T,C},
   v1a/1/-\DISTX/1/{}/{C},
   v1b/5/\DISTX/166/{}/{C},
   p2/6/0/167/{T}/{T,C},
   v3a/7/-\DISTX/168/{}/{T,C},
   v3b/9/\DISTX/257/{}/{T,C},
   p4/10/0/258/{T}/{T,C},
   v5/11/-\DISTX/259/{}/{}} {
  \node (\name) at (-0.2+\position*0.8+\offset,0) [point]{\asserted};
  \node at (-0.2+\position*0.8+\offset,-0.7) {{\scriptsize $\timepoint$}};
  \node[rotate=0, anchor=north] at ($(\name)+(0,1.3)$) {{\completed}};
  }
  
\node at (\PMIN-1.5,1.0) {$\aboxf$};
\node at (\PMIN-1.5,-0.7) {$\zaplus$};
\node at (\PMIN-1.5,0) {$\Amc$};

\node at ($(ldots)+(0,0.5)$) {\dots};
\draw[{-]},thin,gray]	($(ldots)+(0.2,0.5)$) -- ($(v-1)+(0.2,0.5)$); 

\draw[{[-]},thin,gray]	($(v1a)+(-0.2,0.5)$) -- ($(v1b)+(0.2,0.5)$);
\draw[{[-]},thin,gray]	($(v3a)+(-0.2,0.5)$) -- ($(v3b)+(0.2,0.5)$);
\node at ($(rdots)+(0,0.5)$) {\dots};
\draw[{[-},thin,gray]	($(v5)+(-0.2,0.5)$) -- ($(rdots)+(-0.2,0.5)$);

%
%

\end{tikzpicture}
\end{center}
\end{example}

The rewriting from Lemma~\ref{lem:rewritinginfinite} can refer to time instants outside of~$\zaplus$.
However, when we want to evaluate a pure FO formula over the finite structure~$\aboxf$, this is not possible anymore, because the first-order quantifiers must quantify over the domain of~$\aboxf$.
Moreover, since the query~$\rew{\phi}$ can contain metric temporal operators, we need to keep track of the distance between the time points in~$\tem$.
Hence, in the following we assume that~$\aboxf$ is given as a first-order structure with the domain $\NI\cup\{b_1,\dots,b_n\}\cup\zaplus$
and additional predicates $\bit$ and $\sign$ such that $\bit(i,j)$, $1\leq j\leq n$, 
is true iff the $j$th bit of the binary representation of the time stamp~$i$ is~$1$, and $\sign(i)$ is true iff $i$ is non-negative.

Thus, we now consider three-sorted first-order formulas with the three sorts $\NI$ (for objects), $\{b_1,\dots,b_n\}$ (for bits) and $\zaplus$ (for time stamps).
We denote variables of sort $\zaplus$ by $t,t',t''$.
To simplify the presentation, we do not explicitly denote the sort of all variables, but this is always clear from the context.
Every concept name is now accessed as a binary predicate of sort $\NI\times\zaplus$, \eg $A(a,i)$ refers to the fact that individual~$a$ satisfies~$A$ at time point~$i$. Similarly, role names correspond to ternary predicates of sort $\NI\times\NI\times\zaplus$.

In the following we show that 
the expressions $t'\bowtie t$ and even $t'-t\bowtie m$ for some constant~$m$ and ${\bowtie}\in\{\ge,>,=,<,\le\}$ are definable as first-order formulas using the natural order~$<$ on~$\{1,\dots,m\}$.

More specifically, we define $t'-t\bowtie d$, for some constant~$d<\infty$ and ${\bowtie}\in\{\ge,>,=,<,\le\}$, as first-order formulas with build-in predicates $\bit(t,j)$, $1\leq j\leq n$, which is true iff the $j$-th bit of $t$ is $1$, and $\sign(t)$, which is true iff $t\geq 0$.
W.l.o.g., let $d$ be a non-negative integer; otherwise, the formula can be reformulated as $t-t'\bowtie k$, where $0\leq k=-d$.
First, consider the case of equality. $(t'-t= d)$ is true iff
\begin{align*}
\big(\sign(t')\leftrightarrow\sign^{d}(t) \big) \wedge \forall j \, \big(\bit(t',j) \leftrightarrow \bit^{d}(t,j)\big) \land \neg \ovf^d(t),
\end{align*}
where $\sign^{d}(t)$ checks whether~$t+d$ is non-negative,
$\bit^{d}(t,j)$ says what the $j$th bit of the binary representation of $t+d$ is, and
$\ovf^{d}(t)$ detects whether the addition of~$d$ to~$t$ causes an overflow in the $n$-bit.
These three auxiliary predicates are defined inductively as follows:
\begin{align*}
\ovf^{0}(t) &:= \bot \\
\ovf^{d+1}(t)&:=\ovf^{d}(t)\vee \big(\sign^{d}(t)\wedge \forall j.\; \bit^{d}(t,j)\big) \\
\displaybreak[0] \\
\sign^{0}(t)&:=\sign(t) \\
\sign^{d+1}(t)&:= \sign^{d}(t)\vee \big(\forall j. (\exists j'.(j'<j))\leftrightarrow \neg \bit^{d}(t,j)\big) \\
\displaybreak[0] \\
\bit^{0}(t,j)&:=\bit(t,j) \\
\bit^{d+1}(t,j) &:= \sign^{d}(t)\wedge \big(\bit^{d}(t,j) \,\leftrightarrow
\, \exists j'.(j'<j) \land \neg \bit^{d}(t,j')\big) \\
 &\lor \neg\sign^{d+1}(t)\wedge \big(\bit^{d}(t,j) \,\leftrightarrow\, \exists j'. (j'<j) \land \bit^{d}(t,j')\big)
\end{align*}

The remaining cases ${\bowtie}\in\{\ge,>,<,\le\}$ are obtained similarly to the formulas above.
For example, $t< t'$ can be expressed by looking at the signs and the most significant bit in which they differ, formally:
\begin{align*}
  \big(\neg \sign(t)\land \sign(t')\big)  &\\
  \lor \Big( \big(\sign(t)\leftrightarrow \sign(t')\big) &\land
  \exists j. \big( \forall j'. (j' > j) \to (\bit(t',j) \leftrightarrow \bit(t,j))\big)
   \\
   &\land \big(\bit(t',j) \leftrightarrow \sign(t) \big)
  \land \big(\bit(t',j) \leftrightarrow \neg \bit(t,j) \big)
  \Big).
\end{align*}

\begin{restatable}{lemma}{secondlevelrepresentatives}
\label{lem:second-level-representatives}
  For $\rew{\phi}$ there is a constant $N\in\Nbb$ such that,
  for every subformula~$\psi$ of~$\rew{\phi}$, every maximal interval~$J$ in $\Zbb\setminus\bigcup\{[i-N,i+N]\mid i\in\tem\}$, all $k,\ell\in J$, and all relevant tuples~\abf over~\NI, we have $\aboxi,k\models\psi(\abf)$ iff $\aboxi,\ell\models\psi(\abf)$.
\end{restatable}
\begin{proof*}
We are going to prove a more specific statement. Namely, let $N_\psi$ be the sum of all interval bounds of temporal formulas in a subformula~$\psi$  of~$\rew{\phi}$ (except~$\infty$). Consequently, for the proof we consider instead every maximal interval~$J$ in $\Zbb\setminus\bigcup\{[i-N_\psi,i+N_\psi]\mid i\in\tem\}$. 

  We show this by induction on the structure of~$\psi$, but only consider three representative cases; the other cases are similar.
  \begin{itemize}
    \item If $\psi$ is the rewriting of an NCQ, then $N_\psi=0$ and the semantics of~$\psi$ depends only on the interpretation at a single time point. Since $k$ and $\ell$ belong to the same maximal interval in $\Zbb\setminus\tem$, by Lemma~4 in~\cite{BoFoKo-RRML} and the construction of~$\aboxi$, this interpretation behaves in the same way at $k$ and at~$\ell$.

    \item If $\psi$ is of the form $\psi_1\luntil_{[c_1,c_2]}\psi_2$, then $N_{\psi_1}\le N_\psi-c_2$ and $N_{\psi_2}\le N_\psi-c_2$.
    Assume that $\aboxi,k\models\psi(\abf)$. Then there exists $j\in[c_1,c_2]$ such that
	\begin{equation}\label{eq:lem:until-statement}
	 \aboxi,k+j\models\psi_2(\abf) \text{ and }\aboxi,m\models\psi_1(\abf), \text{ for all } m \text{ with } k\le m<k+j.
\end{equation}
    In case that $j=c_1=0$, we have $\aboxi,k\models\psi_2(\abf)$.
    Since $k$ and $\ell$ are farther than $N_\psi\ge N_{\psi_2}$ from the nearest element of~$\zaplus$, by induction we also have $\aboxi,\ell\models\psi_2(\abf)$ and thus $\aboxi,\ell\models\psi(\abf)$ in this case.
    Hence, we can assume in the following that $j\ge c_1>0$, and thus in particular $\aboxi,k\models\psi_1(\abf)$.

    Since both $k+j$ and $\ell+c_2$ are farther than $N_\psi-c_2\ge N_{\psi_2}$ from the nearest element of~$\zaplus$, by induction we have $\aboxi,\ell+c_2\models\psi_2(\abf)$.
    Moreover, since $\aboxi,k\models\psi_1(\abf)$ and $k$ as well as all elements in $[\ell,\ell+c_2]$ are farther than $N_\psi-c_2\ge N_{\psi_1}$ from the nearest element of~$\zaplus$, by induction we have $\aboxi,m\models\psi_1(\abf)$ for all $m$ with $\ell\le m\le\ell+c_2$.
    Hence, $\aboxi,\ell\models\psi(\abf)$.
    \item If $\psi$ is of the form $\psi_1\luntil_{[c_1,\infty)}\psi_2$, then we have a similar situation as above, except that $j$ is not bounded by~$c_2$. We can again assume that $j>0$ and $\aboxi,k\models\psi_1(\abf)$.

    Let $p$ be the maximal element of~$J$.
    If $k+j>p+c_1$, then $k+j>\ell$ and the distance between $\ell$ and $k+j$ must be at least~$c_1$.
    Moreover, by assumption~\ref{eq:lem:until-statement} we have $\aboxi,m\models\psi_1(\abf)$ for all $m$ with $p<m<k+j$.
    Since $\aboxi,k\models\psi_1(\abf)$ and all elements in~$J$ are farther than $N_\psi\ge N_{\psi_1}$ from the nearest element of~$\zaplus$, by induction we also have $\aboxi,m\models\psi_1(\abf)$ for all $m$ with $\ell\le m\le p$.
    Thus, $\aboxi,\ell\models\psi(\abf)$.

    We now consider the remaining case that $k+j\le p+c_1$.
    Then both $k+j$ and $\ell+c_1$ are farther than $N_\psi-c_1\ge N_{\psi_2}$ from the nearest element of~$\zaplus$, and thus by induction we have $\aboxi,\ell+c_1\models\psi_2(\abf)$.
    By similar arguments as above, we obtain $\aboxi,\ell\models\psi(\abf)$.
    \qed
  \end{itemize}
\end{proof*}

Hence, for evaluating subformulas of~$\rew{\phi}$, it suffices to keep track of time points up to~$N$ steps away from the elements of~$\zaplus$; this includes at least one element from each of the intervals~$J$ mentioned in Lemma~\ref{lem:second-level-representatives}, since every element of~$\tem$ is immediately surrounded by two elements of~$\zaplus$.

We exploit Lemma~\ref{lem:second-level-representatives} in the following definition of the three-sorted first-order formula $\trans{\psi}{n}{t}$ that simulates the behavior of~$\psi(\xbf)$ at the \enquote{virtual} time point $t+n$, where $n\in[-N,N]$.
Whenever we use a formula $\trans{\psi}{n}{t}$, we require that $t$ denotes a representative for $t+n$.
Due to our assumption that each maximal interval from $\Zbb\setminus\tem$ is represented by its endpoints (see Example~\ref{ex:period1}), we know that $t$ is a representative for $t+n$ iff there is no element of~$\zaplus$ between $t$ and $t+n$.
We can encode this check in an auxiliary formula:
\[ \repr{n}{t}:=\lnot\exists t'.\,(t+n\le t'<t)\lor(t<t'\le t+n). \]

\begin{example}
In Example~\ref{ex:period1}, time points $1$ and $166$ are representatives for the missing time points $2$--$165$, and we have $\aboxf\models \repr{1}{1}$ (with $N=1$).
However, for $\phi_\Tmc=\lnext \neg\text{T}(x)$, we have $\aboxi,1\models\phi_\Tmc(p_1)$, but $\aboxi,166\not\models\phi_\Tmc(p_1)$, \ie the behavior at $1$ and $166$ differs.
To distinguish this, we need to refer to the \enquote{virtual} time point~$2$ (see the gray circled \enquote{v}s below) that is not included in~$\aboxf$, via the formula $\trans[p_1]{\neg\text{T}(x)}{1}{1}$. By~\Cref{lem:second-level-representatives}, it is sufficient to consider time point~$2$, because this determines the behavior at $3$--$165$ .

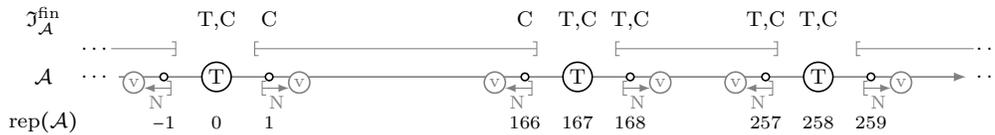
\begin{center}

\newcommand\PMIN{-1}
\newcommand\PMAX{11}
\newcommand\DISTX{0.1}
\begin{tikzpicture}[>=latex, yscale=0.75, xscale=1, semithick,
point/.style={anchor=center,circle,fill=white,draw=black,inner sep=1pt}]\footnotesize

\draw[draw=gray] (\PMIN-0.5,0) -- ++(\PMAX+0.3,0);
\node (ldots) at (\PMIN-0.8,0) {\dots};
\node (rdots) at (\PMAX-0.9,0) {\dots};
%
%

\foreach \name/\position/\offset/\timepoint/\asserted/\completed in 
  {v-1/-1/+\DISTX/-1/{}/{},
   p0/0/0/0/{T}/{T,C},
   v1a/1/-\DISTX/1/{}/{C},
   v1b/5/\DISTX/166/{}/{C},
   p2/6/0/167/{T}/{T,C},
   v3a/7/-\DISTX/168/{}/{T,C},
   v3b/9/\DISTX/257/{}/{T,C},
   p4/10/0/258/{T}/{T,C},
   v5/11/-\DISTX/259/{}/{}} {
  \node (\name) at (-0.2+\position*0.8+\offset,0) [point]{\asserted};
  \node at (-0.2+\position*0.8+\offset,-0.8) {{\scriptsize $\timepoint$}};
  \node[rotate=0, anchor=north] at ($(\name)+(0,1.3)$) {{\completed}};
  }
  
\node at (\PMIN-1.5,1.0) {$\aboxf$};
\node at (\PMIN-1.5,-0.8) {$\zaplus$};
\node at (\PMIN-1.5,0) {$\Amc$};

\node at ($(ldots)+(0,0.5)$) {\dots};
\draw[{-]},thin,gray]	($(ldots)+(0.2,0.5)$) -- ($(v-1)+(0.2,0.5)$); 

\draw[{[-]},thin,gray]	($(v1a)+(-0.2,0.5)$) -- ($(v1b)+(0.2,0.5)$);
\draw[{[-]},thin,gray]	($(v3a)+(-0.2,0.5)$) -- ($(v3b)+(0.2,0.5)$);
\node at ($(rdots)+(0,0.5)$) {\dots};
\draw[{[-},thin,gray]	($(v5)+(-0.2,0.5)$) -- ($(rdots)+(-0.2,0.5)$);

\node (v0) at ($(v-1)+(-0.4,-0.1)$) [circle,fill=white,draw=gray,inner sep=1pt]{\scriptsize\textcolor{gray}{v}};
\node (v1) at ($(v1a)+(0.4,-0.1)$) [circle,fill=white,draw=gray,inner sep=1pt]{{\scriptsize \textcolor{gray}{v}}};
\node (v2) at ($(v1b)+(-0.4,-0.1)$) [circle,fill=white,draw=gray,inner sep=1pt]{\scriptsize\textcolor{gray}{v}};
\node (v3) at ($(v3a)+(0.4,-0.1)$) [circle,fill=white,draw=gray,inner sep=1pt]{{\scriptsize \textcolor{gray}{v}}};
\node (v4) at ($(v3b)+(-0.4,-0.1)$) [circle,fill=white,draw=gray,inner sep=1pt]{\scriptsize\textcolor{gray}{v}};
\node (v6) at ($(v5)+(0.4,-0.1)$) [circle,fill=white,draw=gray,inner sep=1pt]{{\scriptsize \textcolor{gray}{v}}};
\draw[{[->},thin,gray] ($(v-1)+(0.1,-0.2)$) to node[below] {\scriptsize N} ($(v0)+(0.1,-0.1)$);
\draw[{[->},thin,gray] ($(v1a)-(0.1,0.2)$) to node[below] {\scriptsize N} ($(v1)-(0.1,0.1)$);
\draw[{[->},thin,gray] ($(v1b)+(0.1,-0.2)$) to node[below] {\scriptsize N} ($(v2)+(0.1,-0.1)$);
\draw[{[->},thin,gray] ($(v3a)-(0.1,0.2)$) to node[below] {\scriptsize N} ($(v3)-(0.1,0.1)$);
\draw[{[->},thin,gray] ($(v3b)+(0.1,-0.2)$) to node[below] {\scriptsize N} ($(v4)+(0.1,-0.1)$);
\draw[{[->},thin,gray] ($(v5)-(0.1,0.2)$) to node[below] {\scriptsize N} ($(v6)-(0.1,0.1)$);

\end{tikzpicture}
\end{center}

\end{example}

We now define $\trans{\psi}{n}{t}$ recursively, for each subformula~$\psi$ of $\rew{\phi}$.
If $\psi$ is a single rewritten NCQ, then $\trans{\psi}{n}{t}$ is obtained by replacing each atemporal atom $A(x)$ by $A(x,t)$, and similarly for role atoms.
The parameter $n$ can be ignored here, because we assumed that $t$ is a representative for $t+n$, and hence the time points $t$ and $t+n$ are interpreted in \aboxi{} equally.
For conjunctions, we set $\trans{\psi_1\land\psi_2}{n}{t}:=\trans{\psi_1}{n}{t}\land \trans{\psi_2}{n}{t}$ and similarly for the other Boolean constructors.
Finally, we demonstrate the translation for $\luntil$-formulas (the case of $\lsince$-formulas is analogous).
We define $\trans{\psi_1\luntil_{[c_1,c_2]}\psi_2}{n}{t}$ as
\begin{align*}
  &\exists t'.\bigvee_{n'\in[-N,N]}\bigg((t+n+c_1\le t'+n'\le t+n+c_2)\land\repr{n'}{t'}\land\trans{\psi_2}{n'}{t'}\land{} \\
  &\quad \forall t''.\bigwedge_{n''\in[-N,N]}\Big(\big((t+n\le t''+n''<t'+n')\land\repr{n''}{t''}\big)\to \trans{\psi_1}{n''}{t''}\Big)\bigg),
\end{align*}
where $c_2$ may be~$\infty$, in which case the upper bound of $t+n+c_2$ can be removed.

\begin{restatable}{lemma}{rewritingfinite}
\label{lem:rewritingfinite}
  Let $\Kmc = (\Tmc, \Amc)$ be a consistent \dELHbm KB and $\phi$ be an MTNCQ. \\
  Then $\ans(\trans{\rew{\phi}}{0}{t}, \aboxf) = \ans(\rew{\phi}, \aboxi)$.
\end{restatable}
\begin{proof*}
  We show the following claim by induction on the structure of~$\phi$: For all $i\in\zaplus$, all $n\in[-N,N]$, all relevant tuples~\abf, and all MTNCQs~$\phi$ such that if $\aboxf\models\repr{n}{i}$ then
  \[ \aboxf\models\trans[\abf]{\rew{\phi}}{n}{i} \text{ iff } \aboxi,i+n\models\rew{\phi}(\abf). \]
  Since this includes the case where $i\in\tem$, $n=0$, for which $\aboxf\models\repr{0}{i}$ holds, the statement of the lemma follows.

  If $\phi$ is an NCQ, then
  \[ \aboxf\models\trans[\abf]{\rew{\phi}}{n}{i} \text{ iff } \aboxi,i\models\rew{\phi}(\abf) \text{ iff } \aboxi,i+n\models\rew{\phi}(\abf) \]
  since $i$ is a representative for $i+n$ and a single temporal variable~$t$ is used in $\trans{\rew{\phi}}{n}{t}$ to denote \enquote{current} time point in~$\rew{\phi}$.

  For the Boolean constructors, the claim follows immediately from the semantics of first-order logic.

  We now consider a formula of the form $\phi\luntil_I\psi$.
  By induction, we know that the following hold:
  (i) $\aboxf\models\trans[\abf]{\rew{\psi}}{n'}{i'}$ iff $\aboxi,i'+n'\models\rew{\psi}(\abf)$, for any time point $i'$ with $i'+n'\geq i+n$, and (ii) $\aboxf\models\trans[\abf]{\rew{\phi}}{n''}{i''}$ iff $\aboxi,i''+n''\models\rew{\phi}(\abf)$
  for all time points $i''$ and offsets $n''$ such that $i+n\leq i''+n''<i'+n'$
 (assuming w.l.o.g.\ that $\phi$ and $\psi$ have the same answer variables).

  Hence, the formula $\trans[\abf]{\rew{\phi}\luntil_I\rew{\psi}}{n}{i}$ checks the conditions required for the satisfaction of the $\luntil_I$-expression for all time points in $\bigcup\{[i-N,i+N]\mid i\in\zaplus\}$.
  However, Lemma~\ref{lem:second-level-representatives} tells us that, if $\rew{\psi}$ is satisfied in~$\aboxi$ at some time point $i'+n'$ with $n'>N$, then this is also the case for $n'=N$.
  Similarly, to check whether $\rew{\phi}$ is satisfied at all time points between $i+n$ and $i'+n'$, it suffices to consider the time points up to $N$ away from some element of~$\zaplus$.
  Hence, $\aboxf\models\trans[\abf]{\rew{\phi}\luntil_I\rew{\psi}}{n}{i}$ iff $\aboxi,i+n\models(\rew{\phi}\luntil_I\rew{\psi})(\abf)$.
  \qed
\end{proof*}

This lemma allows us to compute in polynomial time that patient~$p_1$ from~\Cref{ex:patients} is an answer to $\phi(x)$ from~\Cref{ex:chemotherapy} exactly at time points~$257$ and $258$.
Below we summarize our tight complexity results, which by Lemma~\ref{lem:model-construction-correctness} also hold for rooted MTCQs under certain answer semantics.

\begin{restatable}{theorem}{complexity}
  Answering rooted MTNCQs under minimal-world semantics over \dELHbm KBs is \ExpSpace-complete, and \PTime-complete in data complexity.
\end{restatable}
\begin{proof*}
  \ExpSpace-hardness is inherited from propositional MTL~\cite{AlHe-JACM94,FuSp-ICTAC08}.
  Moreover, first-order formulas over finite structures can be evaluated in \PSpace \cite{Vard-STOC82}.
  Finally, the size of $\trans{\rew{\phi}}{0}{t}$ is bounded exponentially in the size of~$\phi$ and~\Tmc: Each rewritten NCQ~$\rew{\psi}$ may be exponentially larger than~$\psi$, and each $\trans{\psi_1\luntil_I\psi_2}{n}{t}$ introduces exponentially many disjuncts and conjuncts (but the nesting depth of constructors in this formula is linear in the nesting depth of~$\psi_1\luntil_I\psi_2$).

  For data complexity, hardness is inherited from atemporal \EL \cite{CDD+-AI13}.
  The evaluation of $\text{FO}(<,\bit)$-formulas is in DLogTime-uniform $AC^0$ in data complexity \cite{DBLP:conf/coco/Lindell92}, and the size of our rewriting only depends on the query and the TBox. 
  By Lemmas~\ref{lem:rewritinginfinite} and~\ref{lem:rewritingfinite} and since $\aboxf$ is of size polynomial in the size of~\Amc, deciding whether a tuple~\abf is a minimal-world answer of an MTNCQ w.r.t.\ a \dELHbm KB is possible in~\PTime.
  \qed
\end{proof*}

\section{Related Work}
Until a decade ago, the work on combining ontology and temporal languages was mostly focused on the main reasoning tasks such as concept
subsumption (whether an inclusion between concepts is entailed by a temporal ontology) and knowledge base satisfiability
(whether a given knowledge base consisting of a temporal data instance and a temporal ontology has a model); we refer the reader
to~\cite{DBLP:reference/fai/ArtaleF05,Baader:2003:EDL:885746.885753,LuWZ-TIME08,DBLP:journals/tocl/ArtaleKRZ14} for details.
Particularly, in the presence of a single rigid role, allowing the operator~$\diamondp$ on both sides of \EL CIs makes subsumption undecidable \cite{AKL+-TIME07}.
In \cite{GuJK-IJCAI16}, a variety of restrictions are investigated to regain decidability, such as acyclic TBoxes and restrictions on the occurrences of temporal operators.
In particular, allowing the qualitative operators $\diamondpm,\diamondm,\diamondp,\diamondc$ only on the left-hand side of CIs makes the logic tractable.
Less closely related work considers temporalized DLs where LTL operators are used to combine axioms, but are not allowed within concepts \cite{BaGL-TOCL12}.

Since atomic query answering (whether a concept or role name is entailed by a knowledge base) can be reduced to the satisfiability checking problem, further developments of query languages and the complexities of ontology-mediated query answering over temporal data have recently appeared in the literature.
In the following we briefly summarize recent research on combining ontology and query languages with negation and temporal formalisms.
For a general overview of temporal ontology and query languages, see \cite{LuWZ-TIME08,artale2017ontology}.

In this research we focus on a discrete timeline (over $\Zbb$) and data facts stamped with a single time point.

However, in the literature there are other approaches how to incorporate temporal formalisms in an ontology and the data, \eg dense timelines, like \Qbb or \Rbb~\cite{BKR+-JAIR18,DBLP:conf/ijcai/RyzhikovWZ19}; or interval-based data models, where facts are stamped with a pair of time points denoting the interval in which they are true~\cite{KPP+-IJCAI16,BKR+-JAIR18}.

Discussing this work is beyond the scope of this article.
We only want to mention here that choosing $\Zbb$ rather than $\Qbb$ or $\Rbb$ is not a restriction in our setting since (i) our formalism allows arbitrarily large gaps between time points in the input and (ii) there are no inferences that would require a dense model of time, e.g. computations approaching some time point only in the limit.

Within the discrete time point-based approach, one can distinguish between formalisms with LTL temporal constructs and formalisms employing more refined Metric Temporal Logic (MTL) operators~\cite{DBLP:journals/jacm/AlurFH96}.
MTL extends LTL with temporal intervals for the modal operators, restricting them to a specific time range.

Combining ontology-mediated query answering with LTL operators has been investigated in depth.
In particular, similarly to the query language adapted in this paper, there is a multitude of works~\cite{BaBL-CADE13,DBLP:conf/ijcai/BorgwardtT15,DBLP:conf/ausai/BaaderBL15,BaBL-JWS15} investigating the complexity of answering LTL-CQs that are obtained from LTL formulas by replacing occurrences of propositional variables by arbitrary conjunctive queries (CQs).
Moreover, as in this article, the research~\cite{DBLP:conf/frocos/BorgwardtLT13,borgwardt2015temporalizing} focuses on the rewritability properties of LTL-CQs.
An orthogonal approach for query rewriting over a temporalized \textit{DL-Lite}  ontology was proposed in~\cite{AKWZ-IJCAI13,AKK+-IJCAI15}. Here the focus mainly lies on increasing expressivity of an ontology language by allowing a concept or a role to be prefixed with LTL constructs.
Only recently, also metric variants of LTL-CQs have been considered \cite{GuJO-ECAI16,BBK+-FroCoS17,Thos-KR18}.

Negation in queries with the classical open-world semantics results in non-tractable (mostly \coNP or even undecidable) query evaluation~\cite{Rosa-ICDT07,GIKK-JWS15}. Moreover, prior work~\cite{BoTh-GCAI15,borgwardt2020temporal} on temporalized ontology-mediated query answering with negation shows that the high complexity of temporal query answering with negation is mostly due to the open-world assumption for negation in a query language.
There are several approaches how to introduce negation in ontology-mediated query answering without losing tractability, \eg closed predicates, epistemic semantics, Skolemization. For details, we refer the reader to Section~\ref{sec:ncqs}, where we motivate the semantics adopted in the paper.
Indeed, as we show by this work, by changing semantics for negation, we can apply efficient (in data complexity) algorithms for temporal query answering with negation.

\section{Conclusion}

Dealing with the absence of information is an important and at the same time challenging task. In many real-world scenarios, it is not clear whether a piece of information is missing because it is unknown or because it is false.

EHRs mostly talk about positive diagnoses and it would be impossible to list all the negative diagnoses, \ie the diseases a patient does not suffer from.
Moreover, EHRs and clinical trial criteria contain an inherent temporal component.
We showed that such a setting cannot be handled adequately by existing logic-based approaches, mostly because they do not deal with closed-world negation over anonymous objects.
We introduced a novel semantics for answering metric temporal conjunctive queries with negation and showed that it is well-behaved also for anonymous objects.
Moreover, we demonstrated combined first-order rewritability, which allows us to answer MTNCQs by using conventional relational database technologies.
The data complexity result of~\PTime shows that MTNCQ answering over \dELHb KBs is not inherently more difficult than CQ answering in \EL.
Similarly, the combined complexity does not increase over the lower bound of \ExpSpace inherited from metric temporal logic.

We are working on an optimized implementation of this method with the aim to deal with queries over large ontologies such as SNOMED\,CT.
First results show that this is feasible, in spite of the theoretical complexity of \ExpSpace.%
\footnote{\url{https://github.com/wko/quelk}}
There is already a prototype system that can translate simple clinical trial criteria from clinicaltrials.gov into MTNCQs \cite{XFB+-ODLS19}.
On the theoretical side, we will further develop our approach to also represent numeric information, such as measurements and dosages of medications, which are important for evaluating eligibility criteria of clinical trials \cite{crowe2015designing,BoJi-JBI18}.
Similar to negation, inequalities in CQs also lead to an increase in complexity under certain answer semantics \cite{GIKK-JWS15}. It would be interesting to investigate their behavior under closed-world semantics, also because inequalities are related to counting queries, \eg counting the number of diseases or treatments of a patient.
Another open problem is to extend our results to temporal logics that support \emph{temporal roles}, \ie where temporal operators can also be applied to role names.

\paragraph{Acknowledgments} We thank the anonymous reviewers for their helpful comments on an earlier draft of the manuscript. This work was supported by the DFG grant BA 1122/19-1~(GOASQ), and grant 389792660
as part of TRR~248 (\url{https://perspicuous-computing.science}).

\bibliographystyle{acmtrans}
\bibliography{bibs}
\end{document}